\documentclass[submission, Phys]{SciPost}

\usepackage{amsthm}
\usepackage{mathtools}
\usepackage{physics}

\usepackage[capitalise]{cleveref}

\theoremstyle{plain}
\newtheorem{lemma}{Lemma}
\newtheorem{theorem}[lemma]{Theorem}

\newtheorem{proposition}[lemma]{Proposition}
\theoremstyle{definition}
\newtheorem{definition}[lemma]{Definition}

\newtheorem{remark}[lemma]{Remark}

\usepackage{amsfonts}
\usepackage{dsfont}
\usepackage{bbm}

\newcommand{\mcl}{\mathcal}

\newcommand{\Id}{\ensuremath{\mathbbm{1}}}

\newcommand{\hs}[1][H]{\mcl #1}
\newcommand{\grph}[1][G]{\mathfrak #1}

\newcommand{\field}[1]{\ensuremath{\mathds{#1}}}
\newcommand{\R}{\field{R}}
\newcommand{\C}{\field{C}}
\newcommand{\N}{\field{N}}
\newcommand{\Z}{\field{Z}}

\DeclareMathOperator{\bond}{bond}
\DeclareMathOperator{\borderbond}{\underline{bond}}

\DeclareMathOperator*{\linspan}{span}

\newcommand{\degengeq}{\unrhd}

\DeclarePairedDelimiter{\ceil}{\lceil}{\rceil}

\usepackage{tikz}
\usepackage{standalone}
\usepackage{tikzscale}

\makeatletter
\tikzoption{canvas is xy plane at z}[]{%
  \def\tikz@plane@origin{\pgfpointxyz{0}{0}{#1}}%
  \def\tikz@plane@x{\pgfpointxyz{1}{0}{#1}}%
  \def\tikz@plane@y{\pgfpointxyz{0}{1}{#1}}%
  \tikz@canvas@is@plane
}
\makeatother

\usepackage{xifthen}
\newcommand{\conditionalparentesis}[1][]{%
  \ifthenelse{\isempty{#1}}{}{\!\qty(#1)}
}

\newcommand{\GHZ}[2][]{\operatorname{GHZ}_{#2}\conditionalparentesis[#1]}
\newcommand{\MAMU}[2][]{\operatorname{MaMu}_{#2}\conditionalparentesis[#1]}

\newcommand{\tensor}[1]{\ensuremath{\left\langle #1 \right\rangle}}
\newcommand{\innerprod}[2]{\left\langle #1,\, #2 \right\rangle}
\renewcommand{\expval}[2]{\innerprod{#2}{{#1}{#2}}}

\newcommand{\ie}{i.\,e.\ }
\newcommand{\eg}{e.\,g.\ }

\newcommand{\half}{\frac{1}{2}}

\theoremstyle{plain}
\newtheorem*{main-result}{Theorem}

\usetikzlibrary{external}

\tikzexternaldisable

\usetikzlibrary{positioning,shapes.geometric,3d, calc, arrows,fit,decorations.pathreplacing}
\tikzset{baseline={([yshift=-2]current bounding box.center)}}

\makeatletter
\tikzoption{canvas is xy plane at z}[]{%
  \def\tikz@plane@origin{\pgfpointxyz{0}{0}{#1}}%
  \def\tikz@plane@x{\pgfpointxyz{1}{0}{#1}}%
  \def\tikz@plane@y{\pgfpointxyz{0}{1}{#1}}%
  \tikz@canvas@is@plane
}
\makeatother

\tikzset{pics/GHZZZD/.style n args={1}{
code={%
\node[regular polygon, rotate=30,regular polygon sides=3, draw, rounded corners=0.08cm,
inner sep=0.1cm,fill=black!20] at (0,0) (-triangle){};
\foreach \x in {1,2,3}
  \fill (-triangle.corner \x) circle[radius=0.027cm];

  \node[scale=0.5] at (0,0) {\ensuremath{#1}};

}},
pics/GHZZZZD/.style n args={1}{
code={%
\node[regular polygon, rotate=0,regular polygon sides=4, draw, rounded corners=0.11cm,
inner sep=0.145cm,fill=black!20] at (0,0) (-square){};
\foreach \x in {1,2,3,4}
  \fill (-square.corner \x) circle[radius=0.027cm];

  \node[scale=0.5] at (0,0) {\ensuremath{#1}};

}}
}

\tikzset{pics/GHZZZ/.style n args={1}{
code={%
\node[regular polygon, rotate=30,regular polygon sides=3, draw, rounded corners=0.08cm,
inner sep=0.1cm,fill=black!20] at (0,0) (-triangle){};
\foreach \x in {1,2,3}
  \fill (-triangle.corner \x) circle[radius=0.027cm];

  \node[scale=0.8] at (0,0) {\ensuremath{#1}};

}},
pics/GHZZZZ/.style n args={1}{
code={%
\node[regular polygon, rotate=0,regular polygon sides=4, draw, rounded corners=0.11cm,
inner sep=0.145cm,fill=black!20] at (0,0) (-square){};
\foreach \x in {1,2,3,4}
  \fill (-square.corner \x) circle[radius=0.027cm];

  \node[scale=0.8] at (0,0) {\ensuremath{#1}};

}},
pics/MamuZZZNDNUM/.style n args={3}{
code={%
 \node[regular polygon, regular polygon sides=3,
inner sep=0.095cm,rotate=30] at (0,0) (-tr2){};

\coordinate (-pt11) at ($(-tr2.corner 1)+0.06*(-tr2.corner 2)-0.06*(-tr2.corner 3)$){};
\coordinate (-pt12) at ($(-tr2.corner 1)-0.06*(-tr2.corner 2)+0.06*(-tr2.corner 3)$){};

\coordinate (-pt21) at ($(-tr2.corner 2)+0.06*(-tr2.corner 3)-0.06*(-tr2.corner 1)$){};
\coordinate (-pt22) at ($(-tr2.corner 2)-0.06*(-tr2.corner 3)+0.06*(-tr2.corner 1)$){};

\coordinate (-pt31) at ($(-tr2.corner 3)+0.06*(-tr2.corner 2)-0.06*(-tr2.corner 1)$){};
\coordinate (-pt32) at ($(-tr2.corner 3)-0.06*(-tr2.corner 2)+0.06*(-tr2.corner 1)$){};

\fill(-pt11) []circle[radius=0.025cm];
\fill(-pt12) []circle[radius=0.025cm];
\fill(-pt21) []circle[radius=0.025cm];
\fill(-pt22) []circle[radius=0.025cm];
\fill(-pt31) []circle[radius=0.025cm];
\fill(-pt32) []circle[radius=0.025cm];

\draw[thick](-pt11)--(-pt22)node[midway, left,xshift=0.11cm] {#1};;
\draw[thick](-pt12)--(-pt32)node[midway, above,yshift=-0.11cm,xshift=0.09cm] {#2};
\draw[thick](-pt21)--(-pt31)node[midway,xshift=-0.17cm,yshift=-0.12cm] {#3};;
}},
pics/MamuZZZND/.style n args={3}{
code={%
 \node[regular polygon, regular polygon sides=3,
inner sep=0.095cm,rotate=30] at (0,0) (-tr2){};

\coordinate (-pt11) at ($(-tr2.corner 1)+0.06*(-tr2.corner 2)-0.06*(-tr2.corner 3)$){};
\coordinate (-pt12) at ($(-tr2.corner 1)-0.06*(-tr2.corner 2)+0.06*(-tr2.corner 3)$){};

\coordinate (-pt21) at ($(-tr2.corner 2)+0.06*(-tr2.corner 3)-0.06*(-tr2.corner 1)$){};
\coordinate (-pt22) at ($(-tr2.corner 2)-0.06*(-tr2.corner 3)+0.06*(-tr2.corner 1)$){};

\coordinate (-pt31) at ($(-tr2.corner 3)+0.06*(-tr2.corner 2)-0.06*(-tr2.corner 1)$){};
\coordinate (-pt32) at ($(-tr2.corner 3)-0.06*(-tr2.corner 2)+0.06*(-tr2.corner 1)$){};

\fill(-pt11) []circle[radius=0.025cm];
\fill(-pt12) []circle[radius=0.025cm];
\fill(-pt21) []circle[radius=0.025cm];
\fill(-pt22) []circle[radius=0.025cm];
\fill(-pt31) []circle[radius=0.025cm];
\fill(-pt32) []circle[radius=0.025cm];

\draw[thick](-pt11)--(-pt22)node[scale=0.8, midway, left,xshift=0.08cm] {\ensuremath{#1}};
\draw[thick](-pt12)--(-pt32)node[scale=0.8, midway, above,yshift=-0.08cm,xshift=0.08cm] {\ensuremath{#2}};
\draw[thick](-pt21)--(-pt31)node[scale=0.8, midway, below,yshift=0.08cm, xshift=0.08cm] {\ensuremath{#3}};
}},
pics/MamuZZZNDsm/.style n args={3}{
code={%
 \node[regular polygon, regular polygon sides=3,
inner sep=0.095cm,rotate=30] at (0,0) (-tr2){};

\coordinate (-pt11) at ($(-tr2.corner 1)+0.06*(-tr2.corner 2)-0.06*(-tr2.corner 3)$){};
\coordinate (-pt12) at ($(-tr2.corner 1)-0.06*(-tr2.corner 2)+0.06*(-tr2.corner 3)$){};

\coordinate (-pt21) at ($(-tr2.corner 2)+0.06*(-tr2.corner 3)-0.06*(-tr2.corner 1)$){};
\coordinate (-pt22) at ($(-tr2.corner 2)-0.06*(-tr2.corner 3)+0.06*(-tr2.corner 1)$){};

\coordinate (-pt31) at ($(-tr2.corner 3)+0.06*(-tr2.corner 2)-0.06*(-tr2.corner 1)$){};
\coordinate (-pt32) at ($(-tr2.corner 3)-0.06*(-tr2.corner 2)+0.06*(-tr2.corner 1)$){};

\fill(-pt11) []circle[radius=0.025cm];
\fill(-pt12) []circle[radius=0.025cm];
\fill(-pt21) []circle[radius=0.025cm];
\fill(-pt22) []circle[radius=0.025cm];
\fill(-pt31) []circle[radius=0.025cm];
\fill(-pt32) []circle[radius=0.025cm];

\draw[thick](-pt11)--(-pt22)node[scale=0.4, midway, left,xshift=0.11cm] {\ensuremath{#1}};
\draw[thick](-pt12)--(-pt32)node[scale=0.4, midway, above,yshift=-0.08cm,xshift=0.08cm] {\ensuremath{#2}};
\draw[thick](-pt21)--(-pt31)node[scale=0.4, midway, below,yshift=0.08cm, xshift=0.08cm] {\ensuremath{#3}};
}},
pics/MamuZZZZNDsm/.style n args={4}{
code={%
 \node[regular polygon, regular polygon sides=4,
inner sep=0.145cm,rotate=0] at (0,0) (-tr2){};

\coordinate (-pt11) at ($(-tr2.corner 1)+0.06*(-tr2.corner 2)-0.06*(-tr2.corner 4)$){};
\coordinate (-pt12) at ($(-tr2.corner 1)-0.06*(-tr2.corner 2)+0.06*(-tr2.corner 4)$){};

\coordinate (-pt21) at ($(-tr2.corner 2)+0.06*(-tr2.corner 3)-0.06*(-tr2.corner 1)$){};
\coordinate (-pt22) at ($(-tr2.corner 2)-0.06*(-tr2.corner 3)+0.06*(-tr2.corner 1)$){};

\coordinate (-pt31) at ($(-tr2.corner 3)+0.06*(-tr2.corner 2)-0.06*(-tr2.corner 4)$){};
\coordinate (-pt32) at ($(-tr2.corner 3)-0.06*(-tr2.corner 2)+0.06*(-tr2.corner 4)$){};

\coordinate (-pt41) at ($(-tr2.corner 4)+0.06*(-tr2.corner 3)-0.06*(-tr2.corner 1)$){};
\coordinate (-pt42) at ($(-tr2.corner 4)-0.06*(-tr2.corner 3)+0.06*(-tr2.corner 1)$){};

\fill(-pt11) []circle[radius=0.025cm];
\fill(-pt12) []circle[radius=0.025cm];
\fill(-pt21) []circle[radius=0.025cm];
\fill(-pt22) []circle[radius=0.025cm];
\fill(-pt31) []circle[radius=0.025cm];
\fill(-pt32) []circle[radius=0.025cm];
\fill(-pt41) []circle[radius=0.025cm];
\fill(-pt42) []circle[radius=0.025cm];

\draw[semithick](-pt11) -- (-pt22)node[scale=0.4, midway, above, yshift=-0.09cm] {\ensuremath{#1}};
\draw[semithick](-pt41)--(-pt32)node[scale=0.4, midway, below ,yshift=0.08cm] {\ensuremath{#3}};
\draw[semithick](-pt21)--(-pt31)node[scale=0.4, midway, left,xshift=0.1cm, ] {\ensuremath{#2}};
\draw[semithick](-pt12)--(-pt42)node[scale=0.4, midway, right, xshift=-0.09cm] {\ensuremath{#4}};
}}
}

\tikzset{pics/MamuZZZZNDsmm/.style n args={4}{
code={%
 \node[regular polygon, regular polygon sides=4,
inner sep=0.145cm,rotate=0] at (0,0) (-tr2){};

\coordinate (-pt11) at ($(-tr2.corner 1)+0.06*(-tr2.corner 2)-0.06*(-tr2.corner 4)$){};
\coordinate (-pt12) at ($(-tr2.corner 1)-0.06*(-tr2.corner 2)+0.06*(-tr2.corner 4)$){};

\coordinate (-pt21) at ($(-tr2.corner 2)+0.06*(-tr2.corner 3)-0.06*(-tr2.corner 1)$){};
\coordinate (-pt22) at ($(-tr2.corner 2)-0.06*(-tr2.corner 3)+0.06*(-tr2.corner 1)$){};

\coordinate (-pt31) at ($(-tr2.corner 3)+0.06*(-tr2.corner 2)-0.06*(-tr2.corner 4)$){};
\coordinate (-pt32) at ($(-tr2.corner 3)-0.06*(-tr2.corner 2)+0.06*(-tr2.corner 4)$){};

\coordinate (-pt41) at ($(-tr2.corner 4)+0.06*(-tr2.corner 3)-0.06*(-tr2.corner 1)$){};
\coordinate (-pt42) at ($(-tr2.corner 4)-0.06*(-tr2.corner 3)+0.06*(-tr2.corner 1)$){};

\fill(-pt11) []circle[radius=0.025cm];
\fill(-pt12) []circle[radius=0.025cm];
\fill(-pt21) []circle[radius=0.025cm];
\fill(-pt22) []circle[radius=0.025cm];
\fill(-pt31) []circle[radius=0.025cm];
\fill(-pt32) []circle[radius=0.025cm];
\fill(-pt41) []circle[radius=0.025cm];
\fill(-pt42) []circle[radius=0.025cm];

\draw[thick](-pt11) -- (-pt22)node[scale=0.4, midway, above, yshift=-0.09cm] {\ensuremath{#1}};
\draw[thick](-pt41)--(-pt32)node[scale=0.4, midway, below ,yshift=0.08cm] {\ensuremath{#3}};
\draw[thick](-pt21)--(-pt31)node[scale=0.4, midway, left,xshift=0.1cm, ] {\ensuremath{#2}};
\draw[thick](-pt12)--(-pt42)node[scale=0.4, midway, right, xshift=-0.09cm] {\ensuremath{#4}};
}}
}

\tikzset{pics/lambdST/.style n args={0}{
code={%
\node[regular polygon, regular polygon sides=3, draw, rounded corners=0.1cm,
inner sep=0.1cm,fill=black!10] at (0,0) (-triangle){};
\node[regular polygon, regular polygon sides=3,
inner sep=0.095cm] at (0,0) (-tr2){};
\foreach \x in {1,2,3}{
\draw[red](-tr2.corner \x)--(0,0);
  \fill (-tr2.corner \x) circle[radius=0.03cm];
  }
}}
}

\tikzset{pics/lambdSTLat/.style n args={0}{
code={%
\node[regular polygon, regular polygon sides=3, draw, rounded corners=0.1cm,
inner sep=0.1cm,fill=black!10] at (0,0) (-triangle){};
\node[regular polygon, regular polygon sides=3,
inner sep=0.095cm] at (0,0) (-tr2){};
\foreach \x in {1,2,3}{
\draw[red](-tr2.corner \x)--(0,0);
  \fill (-tr2.corner \x) circle[radius=0.03cm];
  }
}}
}

\tikzset{pics/GHZlat/.style n args={0}{
code={%
\node[regular polygon, regular polygon sides=3, draw, rounded corners=0.1cm,
inner sep=0.1cm,fill=black!10] at (0,0) (-triangle){};
\node[regular polygon, regular polygon sides=3,
inner sep=0.095cm] at (0,0) (-tr2){};
\foreach \x in {1,2,3}
  \fill (-tr2.corner \x) circle[radius=0.03cm];
}},
pics/Mamulat/.style n args={0}{
code={%
 \node[regular polygon, regular polygon sides=3,
inner sep=0.095cm] at (0,0) (-tr2){};
\fill(-tr2.corner 1) [xshift=0.0265cm]circle[radius=0.025cm];
\coordinate (-pt11) at ([xshift=0.0265cm]-tr2.corner 1) {};
\fill[](-tr2.corner 1) [xshift=-0.0265cm]circle[radius=0.025cm];
\coordinate (-pt12) at ([xshift=-0.0265cm]-tr2.corner 1) {};
\fill[](-tr2.corner 2) [xshift=-0.017cm,yshift=0.017cm] circle[radius=0.025cm];
\coordinate (-pt21) at ([xshift=-0.017cm,yshift=0.017cm]-tr2.corner 2) {};
\fill[](-tr2.corner 2) [xshift=0.02cm,yshift=-0.02cm]circle[radius=0.025cm];
\coordinate (-pt22) at ( [xshift=0.02cm,yshift=-0.02cm]-tr2.corner 2) {};
\fill(-tr2.corner 3) [xshift=-0.02cm,yshift=-0.02cm] circle[radius=0.025cm];
\coordinate (-pt31) at ([xshift=-0.02cm,yshift=-0.02cm]-tr2.corner 3) {};

\fill(-tr2.corner 3) [xshift=0.017cm,yshift=0.017cm]circle[radius=0.025cm];
\coordinate (-pt32) at ( [xshift=0.017cm,yshift=0.017cm]-tr2.corner 3) {};

\draw ([xshift=-0.0265cm]-tr2.corner 1) -- ([xshift=-0.017cm,yshift=0.017cm]-tr2.corner 2);
\draw ([xshift=0.0265cm]-tr2.corner 1) -- ([xshift=0.017cm,yshift=0.017cm]-tr2.corner 3);
\draw ([xshift=0.02cm,yshift=-0.02cm]-tr2.corner 2) -- ([xshift=-0.02cm,yshift=-0.02cm]-tr2.corner 3);

}}
}
\tikzset{pics/GHZlatIII/.style n args={0}{
code={%
\node[regular polygon, regular polygon sides=3, draw, rounded corners=0.1cm,
inner sep=0.082cm,fill=black!10] at (0,0) (-triangle){};
\node[regular polygon, regular polygon sides=3,
inner sep=0.075cm] at (0,0) (-tr2){};
\foreach \x in {1,2,3}
  \fill (-tr2.corner \x) circle[radius=0.03cm];
}}}

\tikzset{pics/MamulatII/.style n args={0}{
code={%
 \node[regular polygon, regular polygon sides=3,
inner sep=0.065cm] at (0,0) (-tr2){};
\fill(-tr2.corner 1) [xshift=0.027cm]circle[radius=0.025cm];
\coordinate (-pt11) at ([xshift=0.027cm]-tr2.corner 1) {};
\fill[](-tr2.corner 1) [xshift=-0.027cm]circle[radius=0.025cm];
\coordinate (-pt12) at ([xshift=-0.027cm]-tr2.corner 1) {};
\fill[](-tr2.corner 2) [xshift=-0.017cm,yshift=0.0165cm] circle[radius=0.025cm];
\coordinate (-pt21) at ([xshift=-0.017cm,yshift=0.0165cm]-tr2.corner 2) {};
\fill[](-tr2.corner 2) [xshift=0.02cm,yshift=-0.021cm]circle[radius=0.025cm];
\coordinate (-pt22) at ( [xshift=0.02cm,yshift=-0.021cm]-tr2.corner 2) {};
\fill(-tr2.corner 3) [xshift=-0.02cm,yshift=-0.021cm] circle[radius=0.025cm];
\coordinate (-pt31) at ([xshift=-0.02cm,yshift=-0.021cm]-tr2.corner 3) {};

\fill(-tr2.corner 3) [xshift=0.017cm,yshift=0.017cm]circle[radius=0.025cm];
\coordinate (-pt32) at ( [xshift=0.017cm,yshift=0.017cm]-tr2.corner 3) {};

\draw ([xshift=-0.0265cm]-tr2.corner 1) -- ([xshift=-0.017cm,yshift=0.017cm]-tr2.corner 2);
\draw ([xshift=0.0265cm]-tr2.corner 1) -- ([xshift=0.017cm,yshift=0.017cm]-tr2.corner 3);
\draw ([xshift=0.02cm,yshift=-0.02cm]-tr2.corner 2) -- ([xshift=-0.02cm,yshift=-0.02cm]-tr2.corner 3);

}}
}

\tikzset{pics/MamulatIIBd/.style n args={0}{
code={%
 \node[regular polygon, regular polygon sides=4,
inner sep=0.065cm] at (0,0) (-tr2){};
\fill(-tr2.corner 1) [xshift=0.027cm]circle[radius=0.025cm];
\coordinate (-pt11) at ([xshift=0.027cm]-tr2.corner 1) {};
\fill[](-tr2.corner 1) [xshift=-0.027cm]circle[radius=0.025cm];
\coordinate (-pt12) at ([xshift=-0.027cm]-tr2.corner 1) {};
\fill[](-tr2.corner 2) [xshift=-0.017cm,yshift=0.0165cm] circle[radius=0.025cm];
\coordinate (-pt21) at ([xshift=-0.017cm,yshift=0.065cm]-tr2.corner 2) {};
\fill[](-tr2.corner 2) [xshift=0.02cm,yshift=-0.021cm]circle[radius=0.025cm];
\coordinate (-pt22) at ( [xshift=0.02cm,yshift=-0.021cm]-tr2.corner 2) {};
\fill(-tr2.corner 3) [xshift=-0.02cm,yshift=-0.021cm] circle[radius=0.025cm];
\coordinate (-pt31) at ([xshift=-0.02cm,yshift=-0.021cm]-tr2.corner 3) {};

\fill(-tr2.corner 3) [xshift=0.017cm,yshift=0.017cm]circle[radius=0.025cm];
\coordinate (-pt32) at ( [xshift=0.017cm,yshift=0.017cm]-tr2.corner 3) {};

\draw ([xshift=-0.0265cm]-tr2.corner 1) -- ([xshift=-0.017cm,yshift=0.017cm]-tr2.corner 2);
\draw ([xshift=0.0265cm]-tr2.corner 1) -- ([xshift=0.017cm,yshift=0.017cm]-tr2.corner 3);
\draw ([xshift=0.02cm,yshift=-0.02cm]-tr2.corner 2) -- ([xshift=-0.02cm,yshift=-0.02cm]-tr2.corner 3);

}}
}
\tikzset{
pics/GHZ/.style n args={0}{
code={%
\node[regular polygon, regular polygon sides=3, draw, rounded corners=0.1cm,
inner sep=0.1cm,fill=black!20] at (0,0) (triangle){};
\node[regular polygon, regular polygon sides=3,
inner sep=0.095cm] at (0,0) (-tr2){};
\foreach \x in {1,2,3}
  \fill (-tr2.corner \x) circle[radius=0.03cm];
}},
pics/Mamu/.style n args={0}{
code={%
 \node[regular polygon, regular polygon sides=3,
inner sep=0.095cm] at (0,0) (-tr2){};
\coordinate (-pt11) at ([xshift=0.0265cm]-tr2.corner 1) {};
\node[transform shape,fill,circle,inner sep=0,minimum size=0.05cm] at (-pt11)(-cc11){};
\coordinate (-pt12) at ([xshift=-0.0265cm]-tr2.corner 1) {};
\node[transform shape,fill,circle,inner sep=0,minimum size=0.05cm] at (-pt12)(-cc12){};
 \coordinate (-pt21) at ([xshift=-0.017cm,yshift=0.017cm]-tr2.corner 2) {};
\node[transform shape,fill,circle,inner sep=0,minimum size=0.05cm,rotate=135] at (-pt21)(-cc21){};
\coordinate (-pt22) at ( [xshift=0.02cm,yshift=-0.02cm]-tr2.corner 2) {};
\node[transform shape,fill,circle,inner sep=0,minimum size=0.05cm,rotate=135] at (-pt22)(-cc22){};

\coordinate (-pt31) at ([xshift=-0.02cm,yshift=-0.02cm]-tr2.corner 3) {};
\node[transform shape,fill,circle,inner sep=0,minimum size=0.05cm,rotate=-135] at (-pt31)(-cc31){};
\coordinate (-pt32) at ( [xshift=0.017cm,yshift=0.017cm]-tr2.corner 3) {};
\node[transform shape,fill,circle,inner sep=0,minimum size=0.05cm,rotate=-135] at (-pt32)(-cc32){};

\draw ([xshift=-0.0265cm]-tr2.corner 1) -- ([xshift=-0.017cm,yshift=0.017cm]-tr2.corner 2);
\draw ([xshift=0.0265cm]-tr2.corner 1) -- ([xshift=0.017cm,yshift=0.017cm]-tr2.corner 3);
\draw ([xshift=0.02cm,yshift=-0.02cm]-tr2.corner 2) -- ([xshift=-0.02cm,yshift=-0.02cm]-tr2.corner 3);

}}
}

\pgfmathsetmacro{\latFacdNN}{1.1}

\pgfmathsetmacro{\latFac}{1.08}

\newcommand{\GHZpicture}[2][1]{
  \begin{tikzpicture}[scale=#1]
    \pic[transform shape] at (0,0) {GHZZZ=#2};
  \end{tikzpicture}
}

\newcommand{\MAMUpicture}[4][1]{
  \begin{tikzpicture}[scale=#1]
    \pic[transform shape] at (0,0)  {MamuZZZND={#2}{#3}{#4}};
   \end{tikzpicture}
}

\newcommand{\lambdapicture}[1][1]{
  \begin{tikzpicture}[scale=#1]
    \pic[transform shape,rotate=-90] at (0,0)  {lambdST};
   \end{tikzpicture}
}

\renewcommand{\author}[2][a]{#2\textsuperscript{#1}}
\newcommand{\affil}[2][a]{{\bf #1} #2}

\begin{document}
\begin{center}{\Large \textbf{
Tensor network representations \\from the geometry of entangled states
}}\end{center}

\begin{center}
\author[1]{Matthias Christandl},
\author[1,2,3]{Angelo Lucia},
\author[1,4,5]{P\'{e}ter Vrana},
\author[1,2,*]{Albert~H. Werner}
\end{center}

\begin{center}
\affil[1]{QMATH, Department of Mathematical Sciences, University of Copenhagen,
  Universitetsparken 5, 2100 Copenhagen, Denmark}
  \\
\affil[2]{NBIA, Niels Bohr Institute, University of Copenhagen, \\ Blegdamsvej 17, 2100 Copenhagen,
  Denmark}
  \\
\affil[3]{Walter Burke Institute for Theoretical Physics and Institute for Quantum Information \&
  Matter, California Institute of Technology, Pasadena, CA 91125, USA}
  \\
\affil[4]{Department of Geometry, Budapest University of Technology and Economics,
\\ Egry J\'{o}zsef u. 1., 1111 Budapest, Hungary}
\\
\affil[5]{MTA-BME Lend\"ulet Quantum Information Theory Research Group}
\\
% TODO: corresponding author
* werner@math.ku.dk
\end{center}

% The abstract is in boldface, and should fit in 8 lines.
\section*{Abstract}{\bf%
  Tensor networks provide descriptions of strongly correlated quantum systems
  %with applications ranging from condensed matter physics to cosmology
  based on an underlying entanglement structure given by a
  graph of entangled states along the edges that identify the indices
  of the local tensors to be contracted.
  %
  %Both the structure of the underlying graph
  %and the bond dimension of the entangled states influence the computational
  %cost of contracting these networks.
  %
  Considering a more general setting, where entangled states on
  edges are replaced by multipartite entangled states on faces, allows
  us to employ the geometric properties of multipartite entanglement
  in order to obtain representations in terms of superpositions of
  tensor networks states with smaller effective dimension, leading to
  computational savings.
  %
  %We give examples where the reduction in bond dimension can be arbitrarily
  %large both in terms of the local physical dimension and in terms of the number
  %of subsystems. We illustrate our method with the resonating valence bond state
  %on the kagome lattice.
 }

% include a table of contents (optional)
\vspace{10pt}
\noindent\rule{\textwidth}{1pt}
\tableofcontents\thispagestyle{fancy}
\noindent\rule{\textwidth}{1pt}
\vspace{10pt}

\section{Introduction}
\label{sec:intro}
Taming the exponential growth of complexity with increasing system size presents
one of the major problems in the theory of quantum many-body systems.
Tailor-made Ansatz-classes such as tensor network states have allowed for
tremendous progress over the last two decades both in terms of numerical
\cite{white1992density,verstraete2004renormalization,verstraete2008matrix,Orus2014}
as well as analytical work \cite{Fannes1992,PerezGarcia2007}. This includes
results on ground state properties
\cite{verstraete2006matrix,Perez-Garcia_2008,Schuch2010}, the classification of
quantum phases \cite{Chen2011,Schuch2011}, disordered systems
\cite{bauer2013area,friesdorf2015many,friesdorf2015local,chandran2015spectral,Goldsborough_2017},
the behaviour of open quantum many-body systems
\cite{werner2016positive,kshetrimayum2017simple},
 critical systems \cite{Vidal2008}, as well as related to the AdS/CFT-correspondence \cite{Pastawski2015}.

At the heart of such tensor network approaches is the idea to obtain a class of physical states of interest from an underlying resource state by the application of local linear  operations, which can be seen as applying stochastic local operations and classical communication \cite{dur2000three}. In the case of matrix product states (MPS) and projected entangled pair states (PEPS) these states are given by networks of maximally entangled states. For certain applications, other tensor network structures have been introduced such as tree tensor networks \cite{murg2010simulating,nakatani2013efficient} and the multi-scale renormalization ansatz (MERA) \cite{Vidal2007,evenbly2015tensor}, the latter capturing ground state properties of critical systems.

Another route to generalizing MPS and PEPS, which has been recently explored,
allows for more general resource states beyond EPR-pairs
\cite{Chen2011a,molnar2017generalization,Xie_2014}. These are based on
multi-partite quantum states shared among several lattice sites such as
GHZ-states \cite{molnar2017generalization}. In this work, we further generalize 
this approach by extending both the underlying resource state or
entanglement structure as well as  the class of allowed operations. More precisely, we
allow for one-parameter families of approximate representations, which reproduce
the state of interest to an arbitrary precision.

We show how these approximate representations can be turned into exact
representations in terms of a moderate number of linear superposition of
tensor network states. This approach provides more efficient tensor
network representations for certain classes of states, and gives rise to an
efficient algorithm to reconstruct expectation values faithfully.
In addition, we obtain results that allow to simulate or re-express tensor network states
based on multi-partite resource states in terms of ordinary PEPS, thereby
enabling a numerical treatment of these states by the highly optimized methods
that exist for PEPS. As a concrete example, we show that that semi-injective
PEPS on the two-dimensional square lattice based on GHZ states as introduced in
\cite{molnar2017generalization} with bond dimension $D$ can be represented as a
normal PEPS of bond dimension $2D$.

As an example of the application of our results, we consider the Resonating
Valence Bond (RVB) state, which has originally been proposed as the ground state
of spin liquids \cite{Anderson_1973} and is also of importance in the theory of
high-temperature superconductivity \cite{anderson1987resonating}. The RVB state
has also been studied extensively in the context of PEPS
\cite{verstraete2006criticality,Schuch2012, poilblanc2012topological}. A first
tensor network representation of this state as a PEPS with bond dimension equal
to 3 was introduced in \cite{verstraete2006criticality}. We present two new
representations of the state: a PEPS with non-uniform bond dimensions on the
kagome lattice, with bonds of dimension $(2,2,3)$ depending on the orientation
of the bond, which we show is optimal; and a representation in terms of a
superposition of a linear (in the system size) number of PEPS with bond
dimension equal to 2.

The paper is organized as follows. In
Section~\ref{sec:MPS:and:AlgComplThe}, we revisit the definition of
MPS and connect it to notions in algebraic complexity theory. In
particular, we introduce degenerations as a way of obtaining
approximate state representations with smaller bond dimension. This
leads to the concept of border bond dimension and we give a first
example in terms of the W-state where this approximate representation
leads to a provably more efficient representation.
These ideas are
then generalized in Section~\ref{sec:entanglement-structures} to PEPS
and other entanglement structures of multi-partite states, seen as
representations of graphs and hypergraphs.
In Section~\ref{sec:exStatRep}, we consider the question how to transform
a given entanglement structure into another one based on degenerations
and provide an efficient algorithm to compute exact expectation values
even in this approximate setting. At the same time, this result
lets us interpret states obtained from degenerations as arising from
superpositions of tensor networks states with
the number of superimposed states growing linearly with the system size.

The main building block for this general result turns out to be an
approximate conversion between the plaquette states of the two
entanglement structures involved. Therefore, we present in
Section~\ref{sec:plaquette-conversions} specific examples of such
plaquette conversions between important tensor network classes such as
PEPS, generalized injective PEPS and the RVB state on the kagome
lattice, proving lower and upper bounds on the required bond
dimensions.
In \Cref{sec:comp-complexity} we include an analysis of the
computational cost of computing expectation values of tensor networks
states on the square and kagome lattice using these more efficient
approximate representations, focusing in particular on the RVB state.

\section{Matrix product states \& Algebraic complexity theory}
\label{sec:MPS:and:AlgComplThe}
As a starting point for more general tensor networks, we discuss in this section the concept of MPS representations from the point of view of algebraic complexity theory. In particular, we introduce the concept of degenerations, which correspond to a weaker notion of MPS representations that allows for a controlled approximation error. We then show that this notion leads to a more efficient translation-invariant MPS representation of the W-state on a ring.

Let us first recall the definition of an MPS. To this end, we consider a state vector $T\in\left(\C^d\right)^{\otimes L}$ of $L$ spins of local dimension $d$. Expanding $T$ with respect to a product basis $\{\ket{i_1,\dots, i_l}\}$ we obtain
\begin{align}
  T = \sum_{i_1,\dots, i_L=1}^d T_{i_1,\dots, i_L} \ket{i_1,\dots, i_l}
\end{align}
with $T_{i_1,\dots, i_L}$ denoting the basis coefficients. An MPS representation of $T$ can now be seen as a particular way of decomposing the order $L$ coefficient tensor $T_{i_1,\dots, i_L}$ according to
\begin{align}
  T_{i_1,\dots, i_L} = \tr\left(M_{i_1}^{[1]} \cdots M_{i_L}^{[L]}\right)
\end{align}
with $M_{i_j}^{[j]}$ being a $D\times D$-matrix of sufficiently large dimension $D$, the so-called bond dimension. For each spin $j=1,\dots, L$, we can then define an order $3$ tensor according to $M^{[j]} = \sum_{j=1}^d \sum_{\alpha,\beta =1}^D (M_{i}^{[j]})_{\alpha,\beta} \ketbra{\alpha}{\beta}\otimes\ket{i}$ and by setting $A_j = \sum_{j=1}^d\sum_{\alpha,\beta =1}^D (M_{i}^{[j]})\ketbra{i}{\alpha\beta}$, we can in turn express the state vector $T$ as
 \begin{equation}\label{eq:MPS}
    T = \left(\bigotimes_{j=1}^L A_j \right) \qty(\bigotimes_{k=1}^{L} \Omega^D_{k,k+1}), \quad \Omega^D
    = \sum_{l=1}^D \ket{l,l},
  \end{equation}
which is an MPS representation with periodic boundary conditions and bond dimension $D$ of the state $T$. Note that the two tensor products in \eqref{eq:MPS} are shifted with respect to each other by  half a physical lattice site such that $\Omega_{k, k+1}$ corresponds to a maximally entangled state shared between the lattice sites $k$ and $k+1$, whereas $A_j$ acts on the combined virtual space $\C^D\otimes\C^D$ at lattice site $j$ (see also \Cref{fig:mps}). The important observation about \eqref{eq:MPS} is the fact that we are applying these linear maps $A_j$ locally at each lattice site to an underlying resource state, which in the case of an MPS is given by maximally entangled pair states $\Omega^D$ shared between neighbouring lattice sites if we think of the $L$ spins positioned on a one-dimensional ring.

  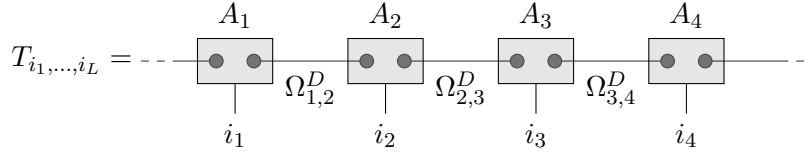
\begin{figure}[t]
    \[ T_{i_1,\dots,i_L} =
      \begin{tikzpicture}[baseline=-2pt]
        \tikzstyle{site}=[shape=circle,draw=black!70,fill=black!50,scale=0.5];
        \tikzstyle{epr}=[-,draw=black!80];
        \tikzstyle{tensor}=[shape=rectangle,draw=black,fill=black, fill opacity=0.1, draw opacity=1, text opacity=1];
      \foreach \x in {0,2,...,6} {
        \node[site] (\x) at (\x,0) {};
      }
      \foreach \x in {1,3,...,7} {
        \node[site] (\x) at ($ (\x,0) - (0.5,0) $) {};
      }
      \draw[epr,dashed] (-1,0) -- (-.5,0);
      \draw[epr] (-.5,0) -- (0);
      \foreach \x/\y/\z/\zz in {1/2/1/2,3/4/2/3,5/6/3/4}{
        \draw[epr] (\x) -- (\y) node[below,midway] (test\x) {$\Omega^D_{\z,\zz}$};
      }

      \draw[epr] (7) -- (7.5,0);
      \draw[epr,dashed] (7.5,0) -- (8,0);

      \foreach \x/\y [count=\c] in {0/1,2/3,4/5,6/7}{
        \node[tensor,label=above:$A_\c$, minimum width=1cm, minimum height=0.6cm] (T\x) at ($(\x)!0.5!(\y)$) {} ;
        \draw[-,black] (T\x) -- +(0,-0.7) node[below] {$i_\c$};
      }

      \end{tikzpicture}
  \]
  \caption{\label{fig:mps}
    Matrix product states as a network of maximally entangled states $\Omega^D$
    shared between physical sites of the 1D lattice to which local operations $A_j$ are applied on combined virtual space on each site.}
\end{figure}

The question of whether a given vector $\psi$ living in a tensor product space $\bigotimes_{j=1}^L \C^{d_j}$ can be transformed into a state $\phi\in\bigotimes_{j=1}^L \C^{d'_j}$ via local linear maps $A_j:\C^{d_j}\mapsto \C^{d'_j}$ is known in the context of algebraic complexity theory as restriction.

\begin{definition}[Restriction]\label{def:restriction}
  Given $\psi \in \bigotimes_{j=1}^m \C^{d_j}$ and $\phi\in\bigotimes_{i=j}^m\C^{d'_j}$ we say that $\psi$ restricts to $\phi$, denoted as $\psi\geq\phi$ if there exist linear maps $\{A_j:\C^{d_j}\mapsto\C^{d'_j}\}$ such that
  \begin{equation}\label{eq:restric}
    \left(\bigotimes_{j=1}^m A_j\right) \psi = \phi\;.
  \end{equation}
\end{definition}
Note that the domain of the local maps $A_j$ is implicitly specified via the chosen tensor decomposition $\bigotimes_{j=1}^m \C^{d_j}$ of the underlying Hilbert space, i.e. each $A_j$ acts on the corresponding tensor factor in this decomposition.

An important generalization of the concept of restriction is that of degeneration. Here, instead of an exact conversion according to \eqref{eq:restric}, we allow for approximate conversions between a state $\psi$ to a state $\phi$ by local operations.

\begin{definition}[Degeneration]\label{def:degeneration}
  Let $\psi \in \bigotimes_{i=1}^m \C^{d_i}$ and $\phi\in\bigotimes_{i=1}^m \C^{d'_i}$ be pure states.
  We say that $\psi$ degenerates to
  $\phi$ with error degree $e$, denoted as $\psi \degengeq^e \phi$,
  if there exist linear maps $A_i(\varepsilon):\C^{d_i}\mapsto\C^{d'_i}$, depending polynomially on $\varepsilon$, such that
  \begin{equation}\label{eq:degener}
    (A_1(\varepsilon)\otimes \cdots \otimes A_m(\varepsilon))\psi = \varepsilon^d \phi +
    \sum_{l=1}^e \varepsilon^{d+l} \widetilde{\phi}_l,
  \end{equation}
  for some tensors $ \widetilde{\phi}_l$ and some integer $d$.
  We simply write $\psi \degengeq \phi$ if $\psi \degengeq^e \phi$ for some error degree $e$.
\end{definition}
\begin{remark}
   It is known (see \eg \cite{Burgisser1997}) that the definition of degeneration as given in \eqref{eq:degener} is equivalent to the following statement: $\psi \degengeq \phi$ if there exists a sequence $\left((A_j(n))_{j=1}^m\right)_n$ of linear maps
$A_i(n):\C^{d_i}\mapsto\C^{d'_i}$ such that
  \[
    \lim_{n\rightarrow\infty}(A_1(n)\otimes \cdots \otimes A_m(n))\psi = \phi\;.
  \]
\end{remark}

\begin{remark}
  The notion of degeneration is strictly weaker than that of restriction, in that given two vectors $\psi \in \bigotimes_{i=1}^m \C^{d_i}$, $\phi \in \bigotimes_{i=1}^m \C^{d'_i}$, $\psi$ can degenerate to $\phi$ even if we cannot find a restriction, \ie $\psi\degengeq\phi$, but $\psi\not\geq\phi$. A well known example of this fact is the degeneration from the GHZ-state on three qubits to the W-state \cite{dur2000three,Acin2001,verstraete2003normal} (see also the W-state example at the end of this section)
\end{remark}

Let us denote by $\GHZ[m]{k}$ the $k$-level Greenberger–Horne–Zeilinger (GHZ) state on $m$ parties:
\begin{equation}
  \GHZ[m]{k} = \sum_{i=1}^k \underbrace{\ket{i} \otimes \dots \otimes \ket{i}}_{m \text{ times}}.
\end{equation}
We note that $\GHZ{k}$ agrees with the unit tensor in algebraic complexity theory, usually  denoted as $\tensor{k}$.
In the cases when $m$ is small, as in the case $m=3$ which we will study extensively, we will use the following graphical notation to
represent the GHZ state:
\[
  \GHZ[3]{k} = \sum_{i=1}^k \ket{i}\ket{i}\ket{i} = \GHZpicture[1.2]{k}.
\]
When the number of parties is clear from the context,  we will simply write $\GHZ{k}$ for simplicity.

Seen as the unit tensor, the GHZ state plays a special role in algebraic complexity theory, which leads us to define the
following quantities.
\begin{definition}[Rank and border rank]
 For $\phi \in \bigotimes_{i=1}^m \C^{d_i}$ we define the \emph{rank} and \emph{border rank} of $\phi$ as
\begin{align}
  \label{eq:rank}
  R(\phi) &= \min\{k\in\N; \; \GHZ[m]{k} \ge \phi\},\\
  \label{eq:border rank}
  \underline{R}(\phi) &= \min\{k\in\N; \; \GHZ[m]{k} \degengeq \phi\},
\end{align}
respectively.
\end{definition}

\begin{remark}\label{remark:rank-sr}
Both the rank and the border rank depend on the tensor product structure of the space where $\phi$
lives: if we regroup the tensor product differently, the rank might change. It is easy to see that if
we group factors together, \ie we see $\phi$ not as an $m$-partite state but as an $m'$-partite
state, with $m'<m$, then both the rank and the border rank will not increase. This is due to the
fact that after regrouping the state $\GHZ[m]{k}$ becomes the state $\GHZ[m']{k}$, so if a
restriction/degeneration to $\phi$ was possible before grouping it will still be possible after
grouping.

Moreover, if $m=2$, then both rank and border rank  of $\phi$ coincide with the Schmidt rank across
the bipartition. Therefore, we can see that the maximal Schmidt rank across any possible
bipartition,
\begin{equation}\label{eq:sr-max}
  \operatorname{Sr}_{\max}(\phi) = \max_{K\subset \{1,\dots, m\}} \rank \tr_{K} \dyad\phi\;,
\end{equation}
is a lower bound to $\underline{R}(\phi)$.
\end{remark}

The question whether a given quantum state $\phi$ can be represented as an MPS with periodic boundary conditions of bond dimension $D$ is equivalent to the question, whether $\bigotimes_{k=1}^{L} \Omega^D_{k,k+1}\geq \phi$, where $\Omega^D_{k,k+1}$ again corresponds to a maximally entangled state with $D$ levels shared between the physical lattice sites $k$ and $k+1$ (see also Fig.~\ref{fig:mps}). The state $\bigotimes_{k=1}^{L} \Omega^D_{k,k+1}$ is known in the context of algebraic complexity theory as the iterated matrix multiplication tensor, which is indeed the
$L$-tensor given by maximally entangled states of dimensions $D_1$, $D_2$, $\dots$, $D_L$  arranged in a
cycle. We will denote this tensor as $\MAMU{D_1,\dots, D_L}$ (for Matrix Multiplication):
\begin{equation}
  \MAMU{D_1,\dots,D_L} = \sum_{i_1,\dots,i_L=1}^{D_1,\dots,D_L} \ket{i_L i_1} \otimes \ket{i_1
    i_2}\otimes \cdots \ket{i_{L-1} i_L}.
\end{equation}
This tensor is often denoted as $\tensor{D_1,\dots,D_m}$ in algebraic complexity. We write $\MAMU[L]{D}$ if $D=D_1=\dots=D_L$, a case typically denoted as $\textrm{IMM}_D^L$ in the literature. As in the case of the GHZ state, we will write $\MAMU{k}$ without the parameter $m$ when this does not cause any ambiguity.

In the cases where $L$ is fixed and small, as for example when $L=3$, we will use a similar graphical notation as for the GHZ-state:
\[
  \MAMU{D_1,D_2,D_3} = \sum_{i_1,i_2,i_3=1}^{D_1,D_2,D_3} \ket{i_1,i_2}\ket{i_2,i_3}\ket{i_3,i_1} =
  \MAMUpicture{D_1}{D_2}{D_3}
\]

As mentioned before, $\MAMU[L]{D}$ restricting to an $L$-tensor $\phi$ is equivalent to
$\phi$ admitting an MPS representation of bond dimension $D$ with periodic boundary conditions. More
generally, since PEPS and other tensor network states are defined in terms of networks of maximally
entangled states, we will be interested in conversions between
$\MAMU[L]{K}$ and other states.
This leads us to define, in analogy to the rank and border rank the following quantities.

\begin{definition}[Bond and border bond dimension]\label{def:bond}
For $\phi \in \bigotimes_{i=1}^m \C^{d_i}$ we define the (MPS) bond dimension and border bond dimension of $\phi$ as
\begin{align}
  \label{eq:bond}
  \bond(\phi) &= \min\{k\in\N; \; \MAMU[m]{k} \ge \phi\},\\
  \label{eq:border bond}
  \borderbond(\phi) &= \min\{k\in\N; \; \MAMU[m]{k}\degengeq \phi\},
  \end{align}
 respectively.
\end{definition}

\begin{remark}\label{remark:bond-sr}
  Note that if we split the vertices $\{1,\dots,m\}$ into $\{1,\dots,k\}$ and $\{k+1,\dots,m\}$ for
  some $k=1,\dots,m$, and we see $\MAMU[m]{k}$ as a bipartite quantum state across this cut, the
  resulting state is equivalent to $\MAMU[2]{k} = \GHZ[2]{k^2}$ (since the $\textrm{MaMu}$ tensor corresponds to periodic boundary conditions).
  Similarly to \eqref{eq:sr-max}, we can consider the maximal Schmidt rank across any
  cut instead of any bipartition (\ie we only consider bipartitions where the two parts are
  contiguous in the spin chain):
  \begin{equation}\label{eq:sr-cut}
  \operatorname{Sr}_{\text{cut}}(\phi) = \max_{k\in \{1\dots m\}} \rank \tr_{1,\dots,k} \dyad\phi.
\end{equation}

Then by the previous argument, we see that
\[
    \operatorname{Sr}_{\text{cut}}(\phi)^\frac{1}{2} \le \borderbond(\phi) \le \bond(\phi).
  \]
\end{remark}

To conclude this section, we present an example where degenerations offer a more efficient state representation, \ie an example where we have a separation between bond and border bond dimension if we require a translation invariant representation in both cases. To this end, consider the $W$-state on $L$ qubits defined as
\begin{equation}
  W(L) = \sum_{\substack{i_1, \dots, i_L = 0\\ i_1 + \cdots + i_L=1}}^1 \ket{i_1, \dots, i_L}\;.
\end{equation}
In the following, we give a translation-invariant representation of $W(L)$ with border bond dimension $2$ independent of the system size $L$. This representation follows immediately from the well-known fact that the W-state (viewed as a homogenous polynomial $x^{L-1} y$) has border-Waring rank equal to two. For completeness we will give the argument below. In contrast to this border bond dimension $2$ representation, the results from \cite{michalek2019quantum} on the quantum Wielandt inequality imply that the  bond dimension of a translation-invariant restriction has to grow as $\exp(\frac{1}{3} \omega(3L))$, with $\omega(x)$ the product logarithm or Lambert function and it has been conjectured that the growth should be of the order $L^{1/3}$  \cite{PerezGarcia2007}. We note however that without the restriction to the translation invariant setting one can also find a bond dimension $2$ representation of the W-state.

In order to find a representation of the W-state $W(L)$ on $L$ parties with border bond dimension $2$, note that
\begin{align}
  \psi_L(\varepsilon) = \begin{pmatrix}1 &0 \\ 0 &\varepsilon\end{pmatrix}^{\otimes L} \ket{0}^{\otimes L} = \left(\ket{0} + \varepsilon\ket{1}\right)^{\otimes L} = \ket{0}^{\otimes L} + \varepsilon W(L) + O(\varepsilon^2)
\end{align}
as a product state has bond dimension $1$. Accordingly, the state $\widetilde{W}(\varepsilon,L) = \psi_L(\varepsilon) - \ket{0}^{\otimes L}$ is  a degeneration from $\MAMU[2]{L}$. The corresponding MPS-matrices of this translation-invariant border bond dimension $2$ representation can be chosen as
\begin{align}
  M_0 = \begin{pmatrix} 1 & 0 \\ 0 & (-1)^{1/L}\end{pmatrix} \quad\text{and}\quad M_1 = \begin{pmatrix}
    \varepsilon &0 \\ 0 &0
  \end{pmatrix}\;,
\end{align}
because  $\tr\left(M_0^L\right) = 0 $ and $\tr\left(M_0^n M_1^m\right) = \varepsilon^m$.

\section{From PEPS to entanglement structures induced by hypergraphs}
\label{sec:entanglement-structures}

Going beyond one spatial dimension, the procedure described for MPS can be generalized to higher dimensional lattices which leads to the notion of projected entangled pair states (PEPS) \cite{verstraete2004renormalization}. Again maximally entangled states are shared with neighbouring lattice sites, and the local operations preparing the state of interest from this underlying resource state are allowed to operate on the combined virtual space that includes all these subsystems. This motivates the following definition of tensor networks and entanglement structures for general graphs.
  \begin{definition}[Entanglement Structure (Graph)]\label{def:es-graph}
    Let $\grph=(V,E)$ be a graph with vertex set $V$, edge set $E$, and let $w$ be an integer-valued weight
    function on the edge set $w:E\to \N$. For each $e\in E$, let $\Omega_{e}\in \C^{w(e)}\otimes \C^{w(e)}$ be the
    maximally entangled state of Schmidt rank $w(e)$.
    An \emph{entanglement structure} or
    \emph{contraction scheme} w.r.t. to $\grph$ is then given by
\begin{equation}
      \Psi(\grph) = \bigotimes_{e\in E} \Omega_{e} \in \bigotimes_{v\in V} \C^{D_v}.
    \end{equation}
    with the local virtual dimension at vertex $v\in V$ given by $D_v=\prod_{e:v\in e} w(e)$. We call the \emph{bond dimension} of $\Psi$ the quantity $\max_{v\in
      V}\{\sqrt[\text{deg}(v)]{D_v}\}$.

    For a fixed integer $D$ we will denote by $\Psi_D(\grph)$ the
    entanglement structure obtained by setting a constant weight $w(e)=D$ on the graph (which will
    then have bond dimension $D$).  We will also say that a state
    $\phi \in \bigotimes_{v \in V} \C^{d_v}$ is representable by $\grph$ with bond dimension
    $D$ iff $\Psi_D(\grph)\geq \phi$, where the locality structure of the restriction maps $A_j$ is given by vertex-set $V$ of $\grph$ according to the tensor decomposition $\bigotimes_{v\in V} \C^{D_v}$.
  \end{definition}

  \begin{remark}\label{remark:plaquettes-are-structures}
  We remark that the term bond dimension is used in two different contexts. In Definition~\ref{def:es-graph}  it refers to the dimensionality of the maximally entangled states that form the graph entanglement structure, whereas in Definition~\ref{def:bond} it characterizes for a given state vector the minimal bond dimension necessary to represent the state as an MPS.
  Note also that the notion of bond dimension and border bond dimension given in \Cref{def:bond} can be naturally extended to
  the case of entanglement structures defined on a general graph, \ie as
  \[ \bond_{\grph}(\phi) = \min\{ D \,|\, \Psi_D(\grph) \ge \psi \},\] and similarly for
  $\borderbond_{\grph}$. Since the tensor $\MAMU[m]{k}$ can also be written as $\Psi_k(\grph[C]_m)$, where $\grph[C]_m$ is
  the cycle graph on $m$ vertices, Definition~\ref{def:bond} of bond dimension and border bond dimension coincide with $\bond_{\grph[C]_m}$ and $\borderbond_{\grph[C]_m}$
\end{remark}

Accordingly, MPS and PEPS fit naturally in this setting with the graph represented given by the path graph
$\grph[L]_L$ in the case of open boundary MPS, by the cycle graph $\grph[C]_L$ in the case of
  periodic boundary MPS, and by a lattice graph in the case of PEPS, respectively (see Figure~\ref{fig:entStruExI}). However, also more general tensor networks
  that allow for example for maximally entangled states between next-to-nearest neighbours can be captured within this framework.

  \Cref{def:es-graph}\ identifies the notion of representability again with the existence of a restriction according to
  Definition~\ref{def:restriction}, where the linear maps $\{A_i\}$ correspond to the local tensors
  defining the tensor network state. We remark that our notion of bond dimension is chosen in such a
  way that it captures how the number of parameters necessary to specify such a tensor network
  state scales with the system size. More precisely, given the bond dimension $D$, the number of
  parameters scales as $\mathcal{O}(\abs{V} D^{\text{deg}(\grph)} d_{\max})$, where $d_{\max}$ is
  the maximal physical dimension given by $\max_{i\in V}(d_i)$ and $\text{deg}({\grph})$ is the
maximal degree of the vertices of $\grph$. This definition is general enough to capture savings in
the bond dimension due to non-uniform edge dimensions with respect to the different edges in the
graph, but at the same time reduces to the usual scaling of $\mathcal{O}(\abs{V} D^2 d_{\max})$ or
$\mathcal{O}(\abs{V} D^4 d_{\max})$ in the case of MPS or PEPS with uniform bond dimension,
respectively.

\begin{figure}[t]
    \centering
 %     \tikzexternalenable
% \tikzsetnextfilename{figure01}
%    \includegraphics[width=\textwidth]{Xentanglstruct.tex}
%     \tikzexternaldisable
      \includegraphics[width=\textwidth]{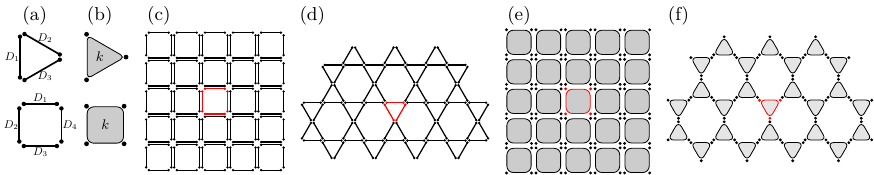}
     \caption{Examples of entanglement structures:
     (a) Plaquette tensors given by maximally entangled
     states shared cyclically between three and four sites, where the indices $D_i$ denote the number of
     levels of the entangled states;
     (b) same as in (a) but with a GHZ-state of $k$ levels shared between the sites;
     (c-d) entanglement structures on the square and kagome lattice, the plaquette shown in red indicates
     how to obtain those from the plaquette states from (a), neighbouring entangled states on the same edge can be reinterpreted as a single maximally entangled state with the number of levels squared;
      (e-f) same lattices as in (c-d) but with a generalized PEPS based on 3- and 4-party GHZ-states. \label{fig:entStruExI}
    }
  \end{figure}

 We will now generalize the concept
  of contraction schemes to representations of hypergraphs, where the underlying entanglement
  structure is given by multipartite entangled states shared among all vertices that are connected
  by a hyperedge.
  \begin{definition}[Entanglement Structure (Hypergraph)]\label{def:ES:Hypergraph}
  Let $\grph=(V,E)$ be a hypergraph, with vertex-set $V$ and hyperedge set $E$. For each $e\in E$,
  let
  $\Omega_e \in \bigotimes_{v\in e} \C^{D_{v,e}}$ be a pure state.
  An \emph{entanglement structure} or \emph{contraction scheme} w.r.t. to $\grph$ is then given by
   \begin{equation*}
    \Psi(\grph) = \bigotimes_{e\in E} \Omega_e \in \bigotimes_{v\in V} \C^{D_v}.
  \end{equation*}
with the local virtual dimension at vertex $v\in V$ given by $D_v = \prod_{e:v\in e} D_{v,e}$ and the
bond dimension of $\Psi(\grph)$ defined as $D=\max_{v\in V}\{\sqrt[\text{deg}(v)]{D}\}$.
\end{definition}
Note that, contrary to the graph case, the hypergraph entanglement structure is not simply defined
by weights on the hyperedges but also by the choice of multi-partite entangled states $\Omega_e$
(since there exist non-equivalent multi-partite entangled states, we cannot simply specify the edge
dimension as in the case of graphs).  As an example, note that the $\GHZ[m]{k}$ can be written as an
entanglement structure on the hypergraph $\grph[H]_m$ with $m$ vertices and a single hyperedge
containing all vertices:
\[ V(\grph[H]_m) = \{0,\dots,m-1\},\quad E(\grph[H]_m) = \{ V\}, \]
by choosing $\Omega_V = \GHZ[m]{k}$.

  In analogy to the graph case, we can still consider a hypergraph entanglement structure $\Psi(\grph)$ as a contraction scheme, with a
  state $\phi$ being representable by $\Psi(\grph)$ iff we can find local maps
  $\{A_v:\C^{D_v}\mapsto d_v\}$ satisfying \eqref{eq:restric} (\ie iff $\Psi(\grph) \geq \phi$).  Note that as in the graph case, the locality structure of the restriction maps $A_j$ is given by vertex-set $V$ of $\grph$ according to the tensor decomposition $\bigotimes_{v\in V} \C^{D_v}$.

  Particular examples of entanglement structures on hypergraphs
  from the literature are projected entangled simplex states \cite{Xie_2014} and semi-injective PEPS  \cite{molnar2017generalization}.
  In the latter case, the vertex set is
  given by the same vertices of the two-dimensional square lattice on $L\times L$ sites (\ie $ \grph[C]_L \times
  \grph[C]_L$), but instead of having an edge for each pair of neighbouring sites, there is instead
  an hyperedge containing the 4 vertices in each of the plaquettes:
  \begin{align*}
    V &= [0,L] \times [0,L] \cap \Z^2, \\
    e \in E \iff e &= \{ (i,j), (i+1,j), (i,j+1), (i+1,j+1)
    \} \qq{for some} (i,j) \in V.
  \end{align*}
  Finally, for each hyperedge $e$, we choose a GHZ state on $4$ parties as $\Omega_e$, so that the
  resulting entanglement structure is given by
  \[
    \Phi = \bigotimes_{e\in E} \GHZ[4]{k}.
  \]
  The bond dimension of $\Phi$ is then simply given by the number of GHZ levels $k$. For $k=2$ this class of states also includes
  the ground state of the CZX-model, exhibiting an on-site $\mathbb{Z}_2$ symmetry \cite{Chen2011a}.

In order to find representations of physical states with optimal bond dimension, we will analyze how
well a given contraction scheme can be expressed in terms of another. To this end, we introduce the
following definition, which specializes \Cref{def:restriction} to the particular case of
entanglement structures.

\begin{definition}[Conversion of entanglement structures]
  Let $\grph$ and $\grph'$ be two graphs or hypergraphs with the same vertex set $V$. Given two
  entanglement structures $\Psi(\grph)$ and $\Psi(\grph')$ we say that $\Psi(\grph)$ restricts to
  $\Psi(\grph')$, and we write $\Psi(\grph) \ge \Psi(\grph')$, if there exist linear maps $A_v:
  \C^{D_v}\to \C^{D'_v}$ for each $v\in V$ such that
  \begin{equation}
    \qty(\bigotimes_{v\in V} A_v) \Psi(\grph) = \Psi(\grph'),
  \end{equation}
  where $D_v$ and $D_v'$ are the local dimension at vertex $v$ of $\Psi(\grph)$ and $\Psi(\grph')$, respectively.
The notion of degeneration specializes to the case of entanglement structures exactly in the same
way as restrictions (\ie by allowing local maps to act according to
the tensor product structure defined by the vertex set $V$).
\end{definition}

\begin{remark}\label{remark:mps-obc}
Note that in the case of a path graph $\grph[L]_L$ on $L$ sites and a graph entanglement structure
$\Psi_k(\grph[L]_L)$ (\ie the entanglement structure of an open boundary condition MPS of bond
dimension $k$), the existence of a degeneration implies the existence of a restriction. More concretely, if
$\Psi_k(\grph[L]_L) \degengeq T$ for some $L$-partite quantum state $T$, then also
$\Psi_k(\grph[L]_L) \ge T$.
This is due to the fact that, by sequential SVD decompositions (see \cite[Theorem
1]{PerezGarcia2007} and \cite[pag. 18-20]{Orus2014}), it is possible to construct $T$ with a
bond dimension equal to the maximal Schmidt rank across any cut $\operatorname{Sr}_{\text{cut}}(T)$
(see \eqref{eq:sr-cut}): equivalently $\Psi_{k}(\grph[L]_L) \ge T$ for
$k=\operatorname{Sr}_{\text{cut}}(T)$.  On the other hand we can repeat the argument of
\Cref{remark:bond-sr} for $\Psi_k(\grph[L]_L)$, but taking into account that we have open boundary
conditions instead: we see that after grouping neighbouring
sites we can convert $\Psi_k(\grph[L]_L)$ to $\Psi_k(\grph[L]_{L'})$ with $L'<L$, so that if
$\Psi_{k}(\grph[L]_L) \degengeq T$ then necessarily $k\ge\operatorname{Sr}_{\text{cut}}(T) $.
On the other hand, as soon as there are cycles in the graph, this argument breaks down, and we have
already seen in the W-state example at the end of Section~\ref{sec:MPS:and:AlgComplThe} that a degeneration can exist when the
corresponding restriction does not.
\end{remark}

\paragraph{RVB state}
Another example of entanglement structure is found in the context of PEPS representations of the Resonating Valence Bond state (RVB) \cite{verstraete2006criticality,Schuch2012, poilblanc2012topological}. Rephrasing the construction used in \cite{Schuch2012} in terms of Definition~\ref{def:ES:Hypergraph}, the RVB state on the kagome lattice can be represented by an entanglement
    structure $\Lambda$, where we assign to each plaquette the 3-party entangled state ${\lambda}\in (\C^3)^{\otimes 3}$ (see Figure~\ref{fig:RVB-state}) given by
\begin{align}
  \lambda = \sum_{i,j,k=0}^2 \epsilon_{i,j,k} \ket{i,j,k} + \ket{2,2,2} = \lambdapicture[1.1]\;,
\end{align}
where $\epsilon_{i,j,k}$ denotes the antisymmetric tensor with
$\epsilon_{0,1,2}=1$.

\begin{figure}[tb]
    \[ \Lambda = \quad
    \begin{tikzpicture}
     \foreach \x in {0,1,2,3,4,5}{
     \pic[transform shape](GHZ1\x) at ($(\latFac*\x,0)+(0,0)$) {lambdST};
     }
     \foreach \x in {0,1,2,3,4,5}{
     \pic[transform shape,rotate=180](GHZ2\x) at ($(\latFac*\x,0.63)+(0,0)$) {lambdST};
     }
     \foreach \x in {0,1,2,3,4}{
     \pic[transform shape](GHZ3\x) at ($(\latFac*\x,0.955)+(0.54,0)$) {lambdST};
     }
      \foreach \x in {0,1,2,3,4}{
     \pic[transform shape,rotate=180](GHZ4\x) at ($(\latFac*\x,1.585)+(0.54,0)$) {lambdST};
     }
      \foreach \x in {0,1,2,3,4,5}{
     \pic[transform shape](GHZ5\x) at ($(\latFac*\x,1.91)+(0,0)$) {lambdST};
     }
      \foreach \x in {0,1,2,3,4}{
     \pic[transform shape,rotate=180](GHZ6\x) at ($(\latFac*\x,-0.32)+(0.54,0)$) {lambdST};
     }

   \node[above=-0.01cm of GHZ31-triangle.side 3,xshift=0.16cm,yshift=-0.08cm]{$\lambda$};
   \node[above=-0.01cm of GHZ32-triangle.side 3,xshift=0.16cm,yshift=-0.08cm]{$\lambda$};
   \node[above=-0.01cm of GHZ33-triangle.side 3,xshift=0.16cm,yshift=-0.08cm]{$\lambda$};

 \end{tikzpicture}
 \]
 \caption{
  The hypergraph entanglement structure $\Lambda$ of the RVB state. The
  triangles $\lambdapicture[0.9]$ represent the $\lambda$ tensor, and the entanglement structure $\Lambda$ is obtained by
  tensoring $\frac{2}{3}L$ copies of it and arranging the vertices according to the kagome lattice.
}
\label{fig:RVB-state}
\end{figure}
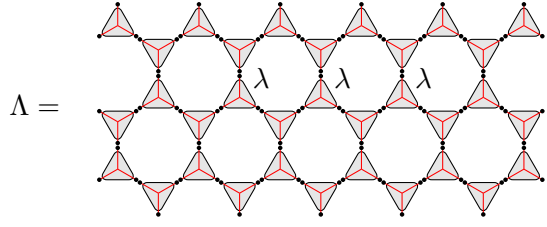

In~\cite{Schuch2012} the entanglement structure $\Lambda$ for the RVB state
composed of plaquette tensors $\lambda$ was shown to have a PEPS representation
with bond dimension 3, by constructing the explicit linear maps realizing the
conversion. The RVB state on the kagome lattice thus has bond dimension at most
3. We now give a representation of the same state with border bond dimension
equal to 2 (in other words we reduce the local virtual dimension at each vertex
from $3^4=81$ to $2^4 =16$), which, as we show in Section~\ref{sec:LambdNoBd2},
is smaller than the optimal PEPS representation that can be obtained with
restrictions.

In order to do so, we need to show that $\MAMUpicture[0.7]{2}{2}{2} \degengeq
\lambdapicture[0.8]$. The MPS matrices of the degeneration are given by
\begin{equation}\label{eq:RVBdegen}
  M^{[j]}_0(\epsilon) = \half\mqty( 0 & \epsilon \\  \epsilon  & 0  )  ,\quad
  M^{[j]}_1(\epsilon) = \mqty( 0 & - \epsilon \\ \epsilon & 0  ) ,\quad
  M^{[j]}_2(\epsilon) = \mqty(  1 & 0  \\ 0 & -1 ) + \frac{\delta_{j,3}\epsilon^2}{2} \mqty(1 & 0 \\ 0 & 1)
\end{equation}
with $j=1,2,3$, resulting in  the state $\varepsilon^2 \lambda + \varepsilon^4 \ket{2}\otimes(\frac{1}{4}\ket{00} - \ket{11})$.

In the following section, we will show how this improved representation can be
used to compute exact contractions and expectation values for the RVB state.

\section{Exact state representations from degenerations}
\label{sec:exStatRep}

Having introduced the concept of tensor network representations in terms of degenerations and border bond dimension in the previous sections, we will now turn to the question how to obtain physical information in this approximate setting.
To this end, we present a general method to turn an approximate conversion between two entanglement structures given by degenerations into an exact one by allowing for superpositions. The main advantage of this approach is on the one hand that this can be accomplished with only a linear overhead in the number of plaquette states involved and on the other hand that it also allows for the computation of exact expectation values.
These properties are summarized in the following theorem. The proof relies on results from algebraic complexity theory \cite{Bini1980,christandl2017tensor} and we include the argument here for the sake of completeness.

\begin{theorem}\label{thm:main}
  Let $\Psi$ and $\Phi$ be the entanglement structures obtained by placing $\psi$ and $\phi$,
  respectively, on $F$ faces of a lattice with $L$ sites. Assume $T$ can be represented by $\Phi$,
  \ie $\Phi \geq T$.
  If  $\psi \degengeq \phi$, then
  \begin{equation}\label{thm:main:I}
    T = \sum_{i=0}^{eF} {W_i},
  \end{equation}
  where each ${W_i}$ can be represented by $\Psi$, \ie $\Psi \ge {W_i}$. The number of terms in the
  representation is linear in $F$, \ie the constant $e$ is only dependent on the degeneration
  $\psi \degengeq \phi$. Moreover expectation values of an observable $O$ under $T$ can be computed
  from expectation values of $2eF+1$ states represented by $\Psi$:
  \begin{equation}\label{thm:main:II}
    \expval{O}{T} = \sum_{i=0}^{2eF} \gamma_i \expval{O}{V_i},
  \end{equation}
  where again each $V_i$ can be represented by $\Psi$, i.e.  $\Psi \geq {V_i} \ \forall i$, and $\gamma_i\in\R$ are known constants depending on $V_i$.
\end{theorem}

\begin{figure}
 \centering
 % \tikzexternalenable
% \tikzsetnextfilename{figure03}
%\includegraphics[width=0.45\textwidth]{XresultPicII.tex}
%    \tikzexternaldisable
  \includegraphics[width=0.45\textwidth]{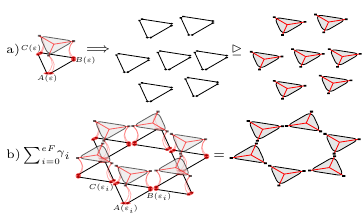}
  \caption{Graphical representation of Theorem \ref{thm:main}: a) A local de\-gen\-eration ($A(\varepsilon)$,$B(\varepsilon)$,$C(\varepsilon)$) depending polynomially on $\varepsilon$  from one plaquette state (pairwise entangled states between three parties) to another ($\lambda$ state), gives rise to a global degeneration between a collection of $F$ plaquette states. b) Evaluating the degeneration at $eF+1$ points $\varepsilon_i$, we can express the full entanglement structure built from the second plaquette state (here $\lambda$ states) as a superposition of $eF+1$ states that arise as restrictions from the first entanglement structure (here pairwise entangled states between three parties). The parameter $e$ is a scaling factor depending on the polynomial degree of the local degeneration, the prefactor $\gamma_i$  is obtained by evaluating the $i$th Lagrange polynomial $\ell_i$ at $0$.
  }
  \label{fig:mainresKago}
\end{figure}

\begin{proof}
  According to Definition~\ref{def:degeneration}, $\psi \degengeq \phi$ if there exist linear maps
$A_i(\epsilon):\C^{d_i}\mapsto\C^{d'_i}$, depending polynomially on the parameter $\epsilon$, such that
  \[
    (A_1(\epsilon)\otimes \cdots \otimes A_m(\epsilon))\psi = \varepsilon^d \phi +
    \sum_{l=1}^e \epsilon^{d+l} \widetilde{\phi}_l,
  \]
  for some multi-partite states $\widetilde{\phi}_l$, approximation degree $d$ and error degree $e$.
We observe that the plaquette degeneration
  $\phi \degengeq \psi$ immediately implies that
  $\phi^{\otimes F} \degengeq \psi^{\otimes F}$, as can be seen by taking the tensor
  product of the local operators given by the degeneration $\phi \degengeq \psi$. As was already
  observed in \cite[Prop. 4]{christandl2017tensor}, the error degree will only grow linearly in the
  number of copies of the degeneration maps, \ie the number of faces $F$ in the lattice and therefore we see that the product of $F$
  copies of $\phi$ degenerates to $F$ copies of $\psi$ with error degree $eF$:
  \begin{equation}\label{eq:deg-copies}
  \underbrace{\psi \otimes \psi \otimes \cdots \otimes \psi}_{F \text{ copies}}
  \degengeq^{e F} \varphi \otimes \varphi
\otimes \cdots \otimes \varphi.
\end{equation}

This degeneration is possible, when all $m$ parties of each of the $F$ copies are considered
independently, \ie when the states in \eqref{eq:deg-copies} are regarded are $m F$-partite
states. In \cite{christandl2017tensor} this was derived in order to show that tensor rank is
strictly submultiplicative under the tensor product.  Note that the degeneration resulting from grouping all the $F$
copies of $\psi$ and $\phi$ into an $m$-tensor was already considered in \cite{Bini1980}, and led to
faster algorithms for matrix multiplication.  In order to prove the theorem, we will
consider instead a different consequence of this argument: grouping the $m F$ tensor factors
according to the underlying lattice, we obtain $\Psi$ and $\Phi$ as $L$-partite states respectively,
which means that
\begin{equation}\label{eq:deg-structures}
  \psi \degengeq^e \phi \implies \Psi \degengeq^{e F} \Phi.
\end{equation}

Similar as in \cite{Bini1980,christandl2017tensor}, we now apply Lagrange interpolation \cite[p. 260]{jeffreys} in order to
transform the degeneration into a restriction. From \eqref{eq:deg-structures}, we can write
\begin{equation}
  \qty(\bigotimes_{i=1}^{F}\left( \bigotimes_{l=1}^m A_l(\varepsilon) \right) )\Psi = \varepsilon^{d} \Phi +
  \sum_{k=1}^{e F} \varepsilon^{d + k} \widetilde{\Phi}_k\,
\end{equation}
for some integer $d$, where the linear maps $A_l(\varepsilon)$, depending polynomially on $\varepsilon$,
are given by copies of the degeneration maps of $\psi \degengeq \phi$ corresponding to the plaquettes $f=1,\dots, F$.
Let $(B_l)_l$ be the local operators given by the restriction $\Phi \ge T$ at the lattice sites $l=1,\dots L$, \ie
$\qty(\bigotimes_{i=1}^L B_l) \Phi = T$. Composing \eqref{eq:deg-structures} with the $(B_l)_l$ and
dividing by $\epsilon^{d}$, we define
\begin{equation}\label{eq:t-eps}
  T(\epsilon) = \epsilon^{-d}  \qty(\bigotimes_{i=1}^{L} B_l ) \qty(\bigotimes_{i=1}^{F}\left( \bigotimes_{l=1}^m A_l(\varepsilon) \right) )\Psi
   =  T +
  \sum_{k=1}^{e F} \varepsilon^{k} \widetilde{T}_k .
\end{equation}
Considering the right hand side, we immediately see that  $T(\epsilon)$ depends polynomially on $\epsilon$ with degree $e F$, and that $T(0)=T$. Moreover, for each $\epsilon \neq 0$, $T(\epsilon)$ is a restriction from
$\Psi$. Evaluating $T(\epsilon)$ at $e F +1$ points $(\epsilon_i)_{i=0}^{e F}$, we can obtain the value at $\epsilon =0 $ via Lagrange interpolation:
\[
T = T(0) = \sum_{i=0}^{e F} \gamma_i T(\epsilon_i) \;,
\]
where $\gamma_i = \ell_i(0)$ is obtained by evaluating the $i$th Lagrange polynomial $\ell_i$
at $0$.
Defining $W_i = \gamma_i T(\epsilon_i)$, we obtain \eqref{thm:main:I}. In order to prove
\eqref{thm:main:II}, we observe that any expectation value with respect to ${T(\epsilon)}$ is given by
\begin{align}\label{eq:expValProof}
  \expval{O}{T(\epsilon)} =
  \expval{O}{T} +  \sum_{k=1}^{e F} \qty( \innerprod{T}{O \tilde T_k}\varepsilon^k +  \innerprod{\tilde
    T_k}{O T} (\overline{\epsilon})^{k} )
  + \sum_{k,k'=1}^{e F} \innerprod{\tilde T_{k'}}{O \tilde T_k} (\overline{\epsilon})^{k} \varepsilon^{k'}.
\end{align}
In case $\varepsilon\in\R$ this is again a polynomial in $\varepsilon$ now of degree $2e F$. Similarly as before, computing $\expval{O}{T(\epsilon)}$ for a fixed $\epsilon\in\R$ amounts to
computing an expectation value for a state $T(\epsilon)$ which has a representation in terms of
$\Psi$. Computing $2e F+1$ of such expectations values is then again sufficient for computing $\expval{O}{T(\epsilon)}$ at
$\epsilon =0$ via interpolation, this proves \eqref{thm:main:II}.
\end{proof}

Let us note that we are not limited to use $\varepsilon\in\R$, but can also choose $\varepsilon\in\C$ if that is more convenient. To this end, let us consider the expression
\[
   \innerprod{T(\overline{\varepsilon})}{O T(\varepsilon)} = \expval{O}{T} +  \sum_{k=1}^{e F} \qty( \innerprod{T}{O \tilde T_k} +  \innerprod{\tilde
    T_k}{O T} )\varepsilon^k
  + \sum_{k,k'=1}^{e F} \innerprod{\tilde T_{k'}}{O \tilde T_k} \varepsilon^{k'+k}.
\]
This is by design again a $2 F$ degree polynomial in $\varepsilon$ with the expectation value $\expval{O}{T}$ in leading order. Hence, computing the scalar product $\innerprod{T(\overline{\varepsilon}_i)}{O T(\varepsilon_i)}$ for $2F+1$ different values of $\varepsilon_i$ will again allow us to compute the value of this polynomial at $\varepsilon=0$ and therefore $\expval{O}{T}$. Alternatively, we could also just insert $\varepsilon$ directly into \eqref{eq:expValProof} and treat $\Re(\varepsilon)$ and $\Im(\varepsilon)$ as independent variables.

We also note that the due to the reduced bond dimension for each of the $V_i$ also the error caused by approximate contraction will be smaller and that by oversampling the number of evaluation points in the degeneration, there is an additional potential for improving the accuracy of the contraction.

Theorem \ref{thm:main} provides degenerations that transform the two entanglement structures into each other plaquette by plaquette. Constructions based on larger units (e.g. several plaquettes) might lead to further reductions in the bond dimension, since the maps on the vertices that are grouped together no longer have a tensor product constraint.

Before looking more generally on plaquette conversions between important classes
of tensor networks in the next section, as a first application of the theorem we
come back to the RVB state on the kagome lattice in terms of the $\Lambda$
entanglement structure introduced at the end of section
\ref{sec:entanglement-structures}. We have presented a degeneration from
$\MAMUpicture[0.8]{2}{2}{2}$ to the plaquette tensors \lambdapicture[0.9] of
$\Lambda$, which has approximation degree $d$ and error degree $e$ both equal to
$2$. Rolling this out on the kagome lattice with $F$ triangles, we obtain a
border PEPS representation of the RVB state of border bond dimension $2$ and
Theorem~\ref{thm:main} then ensures that we can reconstruct the RVB state as a
superposition of $2F+1$ PEPS of bond dimension $2$ or compute expectation values
with $4F+1$ contractions.

\section{Plaquette conversions}
\label{sec:plaquette-conversions}
In this section, we present general strategies and examples for optimized conversion between plaquette states in terms of degenerations. To this end, we consider $m$-tensors, \ie elements of $\bigotimes_{i=1}^m \C^{d_i}$, for
some non-zero integers $(d_i)_i$, which can be equivalently seen as unnormalized pure $m$-partite quantum
states. We will usually consider $m$ to be a small integer (often $m$ will be equal to 3 or 4), as
these $m$-tensors are the building blocks of the entanglement structures we considered in
\Cref{sec:entanglement-structures}. After some definitions and examples that set the scene, we will study the conversion between maximally entangled states shared around circles and GHZ states, which are the basis for conversion between PEPS and more general tensor network states.
To do this, we utilize the correspondence between entangled pairs on the circle and the matrix multiplication tensor (see e.g. \cite{chitambar2008tripartite}).
 This will be first done for 3-party tensors and subsequently for tensors of $m$ parties. In addition,
we prove in Section~\ref{sec:LambdNoBd2} that the MPS representation with bond dimension $(2,2,3)$ for the state $\lambda$, which is the basis for the PEPS representation of the RVB state, is optimal.

\subsection{From \texorpdfstring{$\MAMU[m]{k}$}{MAMU[m](k)} to \texorpdfstring{$\GHZ[m]{k}$}{GHZ[m](k)}: the case \texorpdfstring{$m=3$}{m=3}}

\label{sec:Tri:at}
The aim of this section is to investigate restrictions and degenerations from $\MAMU{k_1,k_2,k_3}$ to
$\GHZ[3]{k}$ and viceversa: this will allow us the express GHZ based hypergraph entanglement structures on triangular lattices as
bond and border bond PEPS representations. In particular, we will prove the following proposition.
\begin{proposition}\label{prop:ghz-bond}
  \begin{equation}\label{eq:ghz-bond-1}
    \frac{1}{2}\qty(1 + \sqrt{4k-3}) < \bond(\GHZ[3]{k}) \le \order{
      k^{\frac{1}{2}+ \frac{c}{\sqrt{\log k}}} },
  \end{equation}
  for some fixed positive $c$.
In other words,
  \begin{equation}\label{eq:ghz-bond-2}
    \MAMU[3]{n} \not\ge \GHZ[3]{n^2-n+1} \qq{and} \MAMU[3]{n} \ge \GHZ[3]{f(n)}
  \end{equation}
  where $f(n) = \order{(n^2)^{1-\frac{c'}{\sqrt{\log n}}}}$ for some positive constant $c'$.
\end{proposition}
However, as it was shown by Strassen in \cite[Thm. 6.6]{strassen1987relative}, there exist degenerations which allow for an
MPS representation of $\GHZ[3]{\ceil{\frac{3}{4}n^2}}$ with border bond dimension $n$. Setting $n=2$, this shows in
particular, that
\[
  \MAMUpicture[1.5]{2}{2}{2} \not\ge \GHZpicture[1.5]{3} \qq{but}
  \MAMUpicture[1.5]{2}{2}{2} \degengeq \GHZpicture[1.5]{3}.
\]
Hence, $\bond(\GHZ[3]{3})>2$, whereas $\borderbond(\GHZ[3]{3})=2$.

Before giving the proof, we discuss a non-symmetric
extension of this result, i.e. degenerations from $\MAMU{k_1,k_2,k_3}$ with
different values of $k_1$, $k_2$, $k_3$.
Following \cite{Vrana2016}, we consider the local diagonal operator
\begin{equation}
  A(\epsilon)\ket{i,j} = \epsilon^{(i-g)^2 + 2 i j}\ket{i,j}
\end{equation}
depending on an integer $g$ which we will fix later. This leads to the transformation
\begin{align*}
( A(\varepsilon)\otimes A(\varepsilon)\otimes A(\varepsilon) ) \MAMU{k_1,k_2,k_3} &=
                                                \epsilon^{2g^2} \sum_{i_1,i_2,i_3=1}^{k_1,k_2,k_3}
                                                \epsilon^{(i_1+i_2+i_3 - g)^2}
                                                \ket{i_1,i_2,}\ket{i_2,i_3}\ket{i_3,i_1}
  \\
                                              &=\epsilon^{2g^2} \sum_{\substack{i_1,i_2,i_3=1\\i_1+i_2+i_3=g}}^{k_1,k_2,k_3}\ket{i_1,i_2}\ket{i_2,i_3}\ket{i_3,i_1} + \order{\epsilon^{2g^2+1}}\;.
\end{align*}
The leading order term in $\epsilon$ corresponds to a GHZ state, because fixing any pair of $i_1$, $i_2$, $i_3$ determines
the third one uniquely. Hence, we only have to determine the number of solutions to the equation
$i_1+i_2+i_3 = g$ for given $n_i$ and inhomogeneity $g$. Choosing $k_1=2$, $k_3=3$ and $k_2=2$ or
$k_2=3$ and $g=5$ then directly leads to
\[
  \MAMUpicture[1.5]{2}{2}{3} \degengeq \GHZpicture[1.5]{4} \qq{and}
  \MAMUpicture[1.5]{2}{3}{3} \degengeq \GHZpicture[1.5]{5}.
\]

These degenerations are optimal, both in the sense that the corresponding restrictions are not
possible, and in the sense that we cannot obtain GHZ states with more levels from a degeneration of
these MaMu tensors. It is also not possible to obtain the same GHZ states from MaMu tensors, where one of the bond dimension is
smaller than the ones we have considered.

We will now turn to the proof of \Cref{prop:ghz-bond}. We will first introduce two definitions and
prove a lemma.

\begin{definition}
  Let $\grph=(V,E)$ be a graph. An \emph{orthogonal representation} of $\grph$ is a mapping
  \[ \pi: V \to \hs \setminus\{0\},\]
  from the graph into some inner product vector space $\hs$ such that
  \begin{equation*}
    (u,v) \in E \implies  \innerproduct{\pi(u)}{\pi(v)}_{\hs} = 0.
  \end{equation*}
  We will denote by $\dim \hs$ the dimension of the orthogonal representation.
\end{definition}

\begin{definition}
Let $\grph[K]_{n,n}$ be the complete bipartite graph on $2n$ vertices, \ie
  \begin{align*}
    V(\grph[K]_{n,n}) &= \{b_0,\dots, b_{n-1}, c_0,\dots, c_{n-1}\},\\
    E(\grph[K]_{n,n}) &= \{ (b_j,c_k) \,|\, j,k=0,\dots,n-1 \}.
  \end{align*}
  Let $\grph[K]^0_{n,n}$ be the graph obtained by removing the edge $(b_0,c_0)$ from
  $\grph[K]_{n,n}$:
  \[ E(\grph[K]^0_{n,n}) =  E(\grph[K]_{n,n})\setminus\{(b_0,c_0)\} = \{ (b_j,c_k) \,|\, j,k=0,\dots,n-1, j\neq k \text{ or } j = k \neq 0 \}. \]
\end{definition}

\begin{lemma}\label{lemma:graph-repr-4}
  With the notation defined above, let  $\pi: \grph[K]^0_{n,n} \to \hs$ be an orthogonal
  representation such that $\dim \hs \le 2(n-1)$. Then at least
  one of the following holds
  \begin{enumerate}
  \item $\dim \linspan \{ \pi(b_i) \,|\, i=1,\dots,n-1 \} < n-1 $,
  \item $\dim \linspan \{ \pi(c_i) \,|\, i=1,\dots,n-1 \} < n-1 $,
  \item $\pi(b_0)$ is orthogonal to $\pi(c_0)$.
  \end{enumerate}
\end{lemma}

\begin{proof}

  Let $\mcl B =  \linspan \{ \pi(b_i) \,|\, i=1,\dots,n-1 \}$ and $\mcl C = \linspan \{ \pi(c_i) \,|\,
  i=1,\dots,n-1 \}$. Since $\pi(b_i)$ is orthogonal to $\pi(c_j)$ for every $i,j=1,\dots,n-1$, we
  have that $\mcl B \perp \mcl C$. If $\mcl B \oplus \mcl C$ is not equal to $\mcl H$, which has
  dimension $\le 2(n-1)$, then at least one of the two has to have dimension strictly smaller than $n-1$,
  so that either \emph{1.} or \emph{2.} holds.
  If not, then $\mcl H = \mcl B \oplus \mcl C$. Since $\pi(b_0)$ is orthogonal to every $\pi(c_i)$
  for $i=1,\dots,n-1$, it is orthogonal to $\mcl C$, and therefore $\pi(b_0) \in \mcl B$. Similarly,
  $\pi(c_0)$ is orthogonal to $\mcl B$ and therefore lies in $\mcl C$. But then $\pi(b_0)$ and
  $\pi(c_0)$ live in orthogonal subspaces and they are themselves orthogonal.
\end{proof}

We are now ready to prove \Cref{prop:ghz-bond}.
\begin{proof}[Proof (\Cref{prop:ghz-bond}).]
  We will start by proving the lower bound of \eqref{eq:ghz-bond-1} as well as first part of
  \eqref{eq:ghz-bond-2}, since they are equivalent as can be seen by setting $k=n^2-n+1$.
  Let us assume that $\GHZ{n^2-n+1}$ has an MPS representation with bond dimension $D\le n$, and let
  us show how to derive a contradiction from this fact. To fix notation, let
  \[ \GHZ{n^2-n+1}  = \sum_{i,j,k=0}^{n^2-n} \tr[A_iB_jC_k] \ket{i,j,k} ,\]
  for some non-zero matrices $\{A_i\}_i, \{B_j\}_j$ and $\{C_k\}_k$ of dimension $D\times D$, such that
  \[ \tr A_i B_j C_k = \begin{cases} 1 \qq{if} i=j=k,\\ 0 \qq{otherwise.}
    \end{cases} \]

  We start by showing that if $D\le n$ we can without loss of generality assume that $A_0$ is
  non-singular. To derive this, we will use the following fact: any linear subspace of $ \mcl M_D$
  containing only singular matrices has dimension at most $D^2-D$ \cite{Dieudonne1948}. Consider
  $S = \linspan\{A_i \,|\, i=0,\dots,n^2-n\} \subset \mcl M_D$.  $S$ is the span of $n^2-n+1$
  matrices: if it contains only singular matrices, then its dimension can be at most $D^2-D$. So if
  $D\le n$, either in $S$ there is one matrix which has full rank or $\dim S \le D^2-D \le n^2-n$,
  which implies that the matrices $(A_i)_i$ are not linearly independent.

  Let $W = (w_{ij})\in U(n^2-n+1)$ a unitary matrix such that $\sum_{i=0}^{n^2-n} w_{0i}A_i$ is either zero or
  full rank.  Then by denoting ${\phi_i} = W\ket{i}$ the rotated basis, we see that
  $(W\otimes W \otimes W) \GHZ{n^2-n+1} = \sum_i {\phi_i}\otimes{\phi_i}\otimes{\phi_i}$ has an MPS
  representation with matrices
  \[ \tilde A_i = \sum_{j} w_{i,j} A_j,\quad\tilde B_i = \sum_{j} w_{i,j} B_j,\quad\tilde C_i =
    \sum_{j} w_{i,j} C_j, \]
  and $\tilde A_0$ is either zero or full-rank. The first case we can exclude, because $\tr \tilde A_0 \tilde
  B_0 \tilde C_0 = 1$. This shows that up to a local unitary on the physical level, we can assume without loss
  of generality that $A_0$ is not singular.

  Let $A_0 = U\Sigma V^*$ be the singular-value decomposition of
  $A_0$. Then $\Sigma >0$ defines a scalar product on $\mcl M_D \simeq \C^{D^2}$ by
  $\innerproduct{X}{Y}_\Sigma = \tr \Sigma X^* Y$. Defining
  \[ \pi(b_j) = B_j^*V, \quad \pi(c_k) = C_kU, \quad j,k=0\dots n^2-n, \]
  we obtain an orthogonal representation of the graph $\grph[K]^0_{n^2-n+1,n^2-n+1}$ (defined in
  \Cref{lemma:graph-repr-4}) on
  $\mcl M_D$ with inner product $\innerproduct{\cdot}{\cdot}_\Sigma$, since
  \[ \innerproduct{\pi(b_j)}{\pi(c_k)}_\Sigma = \tr \Sigma V^* B_j C_k U = \tr A_0 B_j C_k=
    \begin{cases} 1 \qq{if} j=k=0,\\
      0 \qq{otherwise.}
    \end{cases}
  \]
  If $D\le n$, then $\dim \mcl M_D=D^2\le n^2 < n^2 + (n-1)^2 +1 = 2(n^2-n+1)$, which implies that we
  can apply \Cref{lemma:graph-repr-4} and at least one of the conditions stated in it must hold
  true. If \emph{1.} or \emph{2.} hold, then either $\linspan \{B_i\}$ or $\linspan \{C_i\}$ has
  dimension strictly smaller than $n^2-n+1$, but we have already seen that this leads to a contradiction.
  Therefore \emph{3.} must hold, but this also leads to a contradiction: on the one hand we have
  proven that $\tr A_0B_0C_0 = 0$ but we also know that know that
  $\tr A_0B_0C_0 = \innerproduct{\pi(b_0)}{\pi(c_0)} = 1$.

  We will now prove the upper bound of \eqref{eq:ghz-bond-1}. Our starting point is the following result \cite[Thm. 6.6]{strassen1987relative}
  \begin{equation}\label{eq:strassen-deg}
    \MAMU{n} \degengeq^{\gamma n^2}  \GHZ{\left\lceil 3 n^2/4 \right\rceil} ,
  \end{equation}
  for some constant $\gamma > 0$. Let $\alpha$ an integer to be determined later, and consider the
  tensor product of $\alpha$ copies of \eqref{eq:strassen-deg}.  To simplify notation, we set
  $k=(\lceil \frac{3}{4}n^2 \rceil)^\alpha$, so that we get
  \[ \MAMU{n^\alpha} \degengeq^{\alpha \gamma n^2} \GHZ{k} .\]
  As we have discussed previously, it is a well known result in algebraic complexity
  theory that a degeneration can be turned into a restriction by interpolation paying a price in terms of a direct sum (see e.g. \cite{Burgisser1997}). In the present context, this means that we can turn the degeneration into a restriction by supplementing a GHZ state with a number of levels equal to the error degree plus one (see e.g. \cite{christandl2017tensor}).
  Therefore
  we obtain
\[
\GHZ{\alpha\gamma n^2 +1} \otimes \MAMU{n^\alpha} \ge \GHZ{k},
\]
from which follows that
\[
\bond(\GHZ{k}) \le n^\alpha \bond(\GHZ{\alpha\gamma n^2 +1}).
\]
We can trivially bound $\bond(\GHZ{\alpha\gamma n^2 +1})$ by $2\alpha\gamma n^2$,
\begin{equation}\label{eq:ghz-bond-trivial-bound}
\bond(\GHZ{k}) \le 2\alpha\gamma n^{\alpha +2}.
\end{equation}
From the definition of k, we have  $n\le\qty(\frac{4}{3})^{\frac{1}{2}} k^{\frac{1}{2\alpha}}$, and by inserting this into the right-hand side of \eqref{eq:ghz-bond-trivial-bound}, we obtain:
\begin{equation}\label{eq:ghz-bond-good-bound}
  \bond(\GHZ{k}) \le
  2\gamma\alpha \qty(\frac{4}{3})^{1+\frac{\alpha}{2}} k^{\frac{1}{2}+\frac{1}{\alpha}}
  .
\end{equation}

 We now want to choose $\alpha$ in order to minimize the right hand side. We will instead simply
 minimize $\qty(\frac{4}{3})^{\frac{\alpha}{2}}k^{\frac{1}{\alpha}}$, as this will already give the
 right asymptotic scaling. Since the function diverges to infinity when $\alpha$ tends to zero or to
 infinity, we find the minimum by setting the derivative of
 $\frac{\alpha}{2}\log(\frac{4}{3}) + \frac{1}{\alpha} \log k $ to zero:
 \[
   \frac{1}{2}\log(\frac{4}{3}) - \frac{1}{\alpha^2} \log k = 0 \iff \alpha = \alpha^* = \frac{\sqrt{2}}{\log^{1/2}(4/3)} \log^{1/2}(k)
 \]
 Taking $\alpha = \lfloor  \alpha^* \rfloor$, we obtain
 \begin{equation}
   \bond(\GHZ{k}) \le \frac{8}{3}\frac{\sqrt{2}}{\log^{1/2}(4/3)}\gamma k^{\frac{1}{2} +
     \frac{\sqrt{2}\log^{1/2}(4/3)}{\log^{1/2}k} + \frac{\log\log k}{2\log k}}.
 \end{equation}
 Since $\log^{1/2}(k)\geq\log\log k$, we can find $c$ positive as claimed in \eqref{eq:ghz-bond-1}.
 We can improve this bound by minimizing the right hand side of \eqref{eq:ghz-bond-good-bound} instead, obtaining
\begin{align*}
  \bond(\GHZ{k}) &\le \frac{8 \gamma}{3 \log(4/3)}k^{\frac{1}{2} +  \frac{\sqrt{1+2\log(\frac{4}{3})\log(k)}}{\log(k)} + \frac{\log(-1 + \sqrt{1+2\log(\frac{4}{3})\log(k)})}{\log(k)}}\\
  &= \frac{8 \gamma}{3 \log(4/3)} k^{\frac{1}{2}} e^{\sqrt{1+2\log(4/3)\log(k)}}(-1 +\sqrt{1+2\log(4/3)\log(k)})
\end{align*}
Note that the asymptotic scaling of this bound is the same as the one we had obtained by minimizing
$\qty(\frac{4}{3})^{\frac{\alpha}{2}}k^{\frac{1}{\alpha}}$, as we claimed.

 To get the second part of \eqref{eq:ghz-bond-2}, let instead
\[m = \alpha n^{\alpha+2} \le \alpha
 \qty(\frac{4}{3})^{1+\frac{\alpha}{2}} k^{\frac{\alpha+2}{2\alpha}},\]
so that $k\ge\qty(\frac{4}{3})^{-\alpha} \qty(\frac{m}{\alpha})^{\frac{2\alpha}{\alpha+2}}$.
 Then \eqref{eq:ghz-bond-trivial-bound} implies for any $\alpha\geq 1$ that
 \[
   \MAMU{2\gamma m} \ge \GHZ{ \qty(\frac{4}{3})^{-\alpha} \qty(\frac{m}{\alpha})^{\frac{2\alpha}{\alpha+2}} }.
 \]
 Ideally, we would like to take the maximum over $\alpha$ to obtain the best lower bound. Instead, we decide here to maximize the easier function \[ -\alpha \log(4/3) + \frac{2\alpha}{\alpha+2}\log m,\] neglecting the additional summand depending on $-\log\alpha$, since this will already be sufficient to get the desired scaling. The maximum is attained at $\alpha$ satisfying
 \[
   -\log(4/3) + 4 \frac{\log m}{(\alpha+2)^2} = 0 \iff \alpha = \alpha^{**} = 2\frac{\log^{1/2}m}{\log^{1/2}(4/3)} -2,
 \]
 again since the function is smaller or equal to zero for $\alpha$ equal to zero or tending to infinity.
 Since both $(4/3)^{-\alpha}$ and $(m/\alpha)^{\frac{2\alpha}{\alpha+2}}$ are decreasing in
 $\alpha$, substituting $\alpha = \lfloor \alpha^{**} \rfloor$, we obtain that
 \[
   \qty(\frac{4}{3})^{-\alpha} \le \qty(\frac{4}{3})^2 (m^2)^{-\frac{\log^{1/2}(4/3)}{\log^{1/2}(m)}},
   \]
and
 \[
   \qty(\frac{m}{\alpha})^{\frac{2\alpha}{\alpha+2}} \le (m^2)^{
     \qty(1-\frac{\log^{1/2}(4/3)}{\log^{1/2}(m)})
     \qty[ 1 + \frac{1}{\log m}
     \qty(\log(2) - \frac{1}{2} \log\log(4/3) + \frac{1}{2}\log\log m) ]}
 \]
 which implies
\begin{equation}
   \MAMU{2\gamma m} \ge \GHZ{q} \qq{with} q = { c\, (m^2)^{1 - 2\frac{\log^{1/2}(4/3)}{\log^{1/2}m}} },
 \end{equation}
 for some positive constant $c$ which gives the scaling claimed in \eqref{eq:ghz-bond-2}.
\end{proof}

\subsection{From \texorpdfstring{$\MAMU[m]{k}$}{MAMU[m](k)} to \texorpdfstring{$\GHZ[m]{k}$}{GHZ[m](k)}: the general case}\label{sec:MamuToGHZGen}

In this section, we will outline a method to obtain explicit degenerations from $\MAMU[m]{k}$ to $\GHZ[m]{k}$, generalizing some of the results for the $3$-party case from the previous section to
$m$ parties. We will then work out in detail the case for $m=4$ as an example. We leave the reverse bound as an open problem.

As before, $\MAMU[m]{k}$ will denote a network of $m$ parties arranged on a circle each sharing a maximally entangled state with $k$ levels with each of its two nearest neighbours. The goal is to find a local linear transformation $A_l(\varepsilon)$ at each vertex $l$ depending polynomially on $\varepsilon$ such that the leading contribution in $\varepsilon$ of the
resulting state is an $m$-party GHZ state with $k'$ levels
\[
  \qty(\bigotimes_{l=1}^m A_l(\epsilon) ) \sum_{i_1,\dots,i_m=1}^k \ket{i_1
    i_2}\ket{i_2i_3}\dots\ket{i_m i_1} = \epsilon^d\, \GHZ[m]{k'} + \order{\epsilon^{d+1}}\,
\]
where $k'$ should be as large as possible and the kets indicate the grouping of parties. Following \cite{Vrana2016}, we choose the operators $A_l(\epsilon)$ diagonal in the local
product basis, i.e. $A_l(\varepsilon)\ket{i,j} = \varepsilon^{p_l(i,j)}\ket{i,j}$. In addition,
we require that the leading order contribution in $\varepsilon$ is given by those vectors $\ket{i_1,i_2}\cdots\ket{i_m,i_1}$, that satisfy
a certain system of linear equations, \ie $\sum_l c_l i_l = g$ with coefficients vector $c_l$ and
inhomogeneity $g$ belonging to $\Z^\nu$ for some integer $\nu$ \cite{Vrana2016}. This last condition is equivalent to the requirement that the vector $\sum_l c_l i_l - g$ is the zero vector, which can be reexpressed by the norm condition
\begin{align}\label{eq:quadrCond}
0=\braket{\sum_l c_l i_l - g}=\sum_{l=0}^{m}\left(\braket{c_l}{c_l} i_l^2 - (\braket{g}{c_l}+\braket{c_l}{g}) i_l\right) + \braket{g}{g} + \sum_{l\neq l'} \braket{c_l }{c_{l'}} i_l i_{l'}\,.
\end{align}
However, we have to connect this expression back to the local operations $A_l(\varepsilon)$. Indeed, we have to ensure that \eqref{eq:quadrCond} can be generated by a product of local degenerations of the form
\[
  A_l(\epsilon)\ket{i j} = \epsilon^{p_l(i,j)} \ket{i j}\;,
\]
namely $\sum_{l} p_l(i_l,i_{l+1}) = d + \norm{\sum_l c_l i_l - g}_2^2$, which can always be achieved
for all the terms in \eqref{eq:quadrCond} that depend at most on a single index $l$. However, for
the cross-terms this requires $\braket{c_l}{c_{l'}}=0$ if $\abs{l-l'}>1$, forcing the vectors $c_l$
into an orthogonal representation of the cycle graph (giving a lower bound on $\nu$), in which case
we obtain
\begin{align}
\qty(\bigotimes_{l=1}^m A_l(\epsilon) ) \sum_{i_1,\dots,i_m=1}^k \ket{i_1
  i_2}\ket{i_2i_3}\dots\ket{i_m i_1} =
\sum_{i_1,\dots,i_m=1}^k \epsilon^{ \braket{\sum_l c_l i_l - g}}  \ket{i_1i_2}\ket{i_2i_3}\dots\ket{i_m i_1}\;.
\end{align}
Furthermore, we have to ensure that the leading contribution, given by
\begin{align}\label{eq:leadcontr}
  \sum_{\substack{i_1,\dots i_m=1\\ \sum_l c_l i_l = g}}^m \ket{i_1,i_2}\cdots\ket{i_m,i_1}
\end{align}
is indeed locally unitarily equivalent to a  GHZ state, \ie consists of an equal weight superposition
of product states $\psi_r = \psi_{r,1}\otimes \cdots\otimes{\psi_{r,m}}$, such that
$\braket{\psi_{r,l}}{\psi_{r',l}}=\delta_{r,r'}$. Since \eqref{eq:leadcontr} is a superposition of
vectors of the form $\ket{i_1,i_2}\cdots\ket{i_m,i_1}$ this means that fixing a pair of indices
$i_{l'}, i_{l'+1}$ at any vertex $l$ the linear equation $\sum_l c_l i_l = g$ must have at most one
unique solution in the remaining $i_l$. One way of ensuring this is to choose the vectors $c_l$, $c_{l'}$ linearly independent, whenever $\abs{l-l'}>1$.  In other words, we have to choose the vectors $(c_l)_l$ in such a way that if we remove any subset of vectors that share a vertex, the remaining ones have to be linearly independent. The maximal dimension of the GHZ state we can extract is then given by the number of integer solutions to the equation
\begin{align}\label{eq:lineq}
   \sum_{l=0}^m c_l i_l = g\,,
\end{align}
where we optimize over the inhomogeneity $g$. One can get a bound on the number of these solutions
by a probabilistic argument with respect to the inhomogeneity $g.$ However, in order to talk about
the finite $m$ case, we are going write down an explicit expression for \eqref{eq:lineq} that
satisfies all the necessary properties, \ie $\braket{c_l}{c_{l'}}=0$ for $l'\notin\{l-1,l,l+1\}$ and
$\{c_l\}_{l=0}^m\setminus\{c_j,c_{j+1}\}$ linearly independent for all $j$. We define the equations inductively starting from the four-party case
\begin{align}
  \begin{pmatrix}
    1\\1
  \end{pmatrix} i_1 +
  \begin{pmatrix}
    -1\\0
  \end{pmatrix} i_2 +
  \begin{pmatrix}
    1\\-1
  \end{pmatrix} i_3 +
  \begin{pmatrix}
    0\\1
  \end{pmatrix} i_4  = g
\end{align}
Now adding a new vertex and edge into the cycle between $i_4$ and $i_1$ means that now $c_4$ has to be orthogonal to $c_1$ and the new $c_5$ should be orthogonal to all vectors except $c_1$ and $c_4$. This can be achieved by the choice

\begin{align}
  \begin{pmatrix}
    1\\1\\1
  \end{pmatrix} i_1 +
  \begin{pmatrix}
    -1\\0\\0
  \end{pmatrix} i_2 +
  \begin{pmatrix}
    1\\-1\\0
  \end{pmatrix} i_3 +
  \begin{pmatrix}
    0\\1\\-1
  \end{pmatrix} i_4 +
  \begin{pmatrix}
    0\\0\\1
  \end{pmatrix} i_5 = g\;.
\end{align}
This procedure can be repeated leading to the following linear system for the $k$-cycle
\begin{align}
  \begin{pmatrix}
    1 &-1 &1 &0 & 0 & \cdots &0& 0\\
    1 &0 & -1 & 1 & 0 &        &0 &0\\
    1 &0 & 0& -1& 1 & & 0 \\
    1 & 0 &0 &0 &-1\\
    \vdots &\vdots &\vdots &\vdots &\vdots&&\vdots&\vdots\\
    1 & 0& 0 & 0&0& & 1&0\\
    1 & 0& 0 & 0&0& & -1&1
  \end{pmatrix}\cdot \vec{i} = g\;.
\end{align}
In order to find the integer solutions to this problem, we employ the Smith normal form of the matrix on the left-hand side, which gives the general solution vector
\begin{align}\label{eq:iiindCyc}
\vec i=
  \begin{pmatrix}
    z_1+z_2  + A_1\\
    (m-2) (z_1-z_2) + z_2 + A_2 \\
    (m-3)( z_1 - z_2) + z_2 \\
     \vdots\\
    2 (z_1-z_2)+z_2 + A_{m-2}\\
    z_1\\
    z_2
  \end{pmatrix}\,,
\end{align}
where $z_1,z_2$ are arbitrary integers and the constants $(A_l)$ depend on the choice of $g$ by a simple linear integer transformation given by the Smith normal form. In order to obtain the relevant solutions for our specific problem, we have to impose the upper and lower bounds $0$ and $k-1$ if the original maximally entangled states are of dimension $k$ for each entry of the solution vector $\vec i$.

\paragraph{The case $m=4$}
In the case $m=4$, \eqref{eq:iiindCyc} leads to the inequalities

\begin{align*}
0\leq \begin{pmatrix}
    z_1+z_2  + g_2\\
   2z_1-z_2 + g_2-g_1 \\
    z_1\\
    z_2
  \end{pmatrix}\leq n-1 \;.
\end{align*}
Choosing $g_2=g_1\in\{\frac{k}{2},\frac{k-1}{2}\}$ depending on whether $k$ is even or odd leads to the lower bound
 on the number of solutions of the form $\frac{k^2+1}{2}$ for odd dimensions and $\frac{k^2}{2}$ for even $k$.

 This shows $\MAMU[4]{k} \degengeq \GHZ[4]{\ceil{\frac{k^2}{2}}}$, \ie that we can locally degenerate from a cycle of four maximally entangled states with $k$ levels to a four party GHZ state of $\ceil{\frac{k^2}{2}}$ levels.  Hence on the level of plaquette states, we can degenerate from pairwise maximally entangled states  on four parties with $\ceil{\sqrt{2 D}}$ levels to a GHZ state on four parties of $D$ levels. Taking into account that in a two-dimensional square lattice the bond dimension of neighbouring plaquette states have to be combined (see Figure~\ref{fig:entStruExI} (c) and (e)), this means that semi-injective PEPS on the two-dimensional square lattice based on GHZ states as introduced in \cite{molnar2017generalization} with bond dimension $D$ can be represented as a normal PEPS of bond dimension $2D$. By our theorem, expectation values for these generalized PEPS can hence be computed from expectation values of normal PEPS, for which highly optimized numerical codes exist.

 \subsection{Bond dimension of \texorpdfstring{$\lambda$}{lambda} is strictly larger than 2}\label{sec:LambdNoBd2}
 As discussed in Section~\ref{sec:entanglement-structures}, in \cite{Schuch2012} the PEPS representation of the RVB-state is obtained via the multipartite entangled state
 \begin{align}
   \lambda = \sum_{i,j,k=0}^2 \varepsilon_{i,j,k} \ket{i,j,k} + \ket{2,2,2} = \lambdapicture[1.1]\;,
 \end{align}
 with $\varepsilon$ denoting the completely antisymmetric tensor such that $\epsilon_{0,1,2}=1$.
 In \cite{Schuch2012}, the state $\lambda$ was obtained as a restriction from
$\MAMUpicture[0.8]{3}{3}{3}$, obtaining the same PEPS representation of the RVB state with bond dimension
3 from \cite{verstraete2006criticality}. It turns out that is sub-optimal: the tensor $\lambda$ can be obtained also as a restriction from $\MAMUpicture[0.8]{3}{2}{2}$, using the following MPS representation:
\begin{align}\label{eq:RVBopRestr}
  M^{[1]}_0 &=\half\mqty( 0 & 1 & 0 \\ 1 & 0 & 0 ) &
  M^{[1]}_1 &=\mqty( 0 & -1 & 0 \\ 1 & 0 & 0 ) &
  M^{[1]}_2 &=\mqty( 1 & 0 & 1 \\ 0 & -1 & 0 ) \\
  M^{[2]}_0 &=\half\mqty( 0 & 1  \\ 1 & 0 \\ 0 & 0 ) &
  M^{[2]}_1 &=\mqty( 0 & -1  \\ 1 & 0\\ 0 & 0 ) &
  M^{[2]}_2 &=\mqty( 1 & 0 \\ 0 & -1  \\ 1 & 0 ) \\
  M^{[3]}_0 &=\half\mqty( 0 & 1  \\ 1 & 0  ) &
  M^{[3]}_1 &=\mqty( 0 & -1  \\ 1 & 0  ) &
  M^{[3]}_2 &=\mqty( 1 & 0 \\ 0 & -1 ).
\end{align}
This leads to a PEPS representation of the RVB state where the bond dimension is reduced from 3 to 2
on two of the edges of each triangle of the kagome lattice.

 We now prove that this representation is optimal, \ie that $\lambda$ cannot represented as an MPS of bond dimension $2$. This shows in particular a separation of bond and border bond PEPS representations on the kagome lattice, as the PEPS representation for $\lambda$ obtained in \eqref{eq:RVBdegen} has border bond dimension equal to ($2$,$2$,$2$).
\begin{proposition}
   $\MAMUpicture[1.5]{2}{2}{2} \not \geq \lambdapicture[1.6]$
\end{proposition}
\begin{proof}
  Given the general form of an MPS, we have to show that there exists no triples of $2{\times}2$-matrices $(A_i)$, $(B_j)$, $(C_k)$ satisfying
  \begin{align}\label{eq:trLambda}
    \tr(A_i B_j C_k) = \varepsilon_{i,j,k} + \delta_{2,i,j,k}\;.
  \end{align}
  We first note, that the trace on the left hand side gives rise to the usual MPS gauge freedom,
  were we can substitute $A_i\mapsto X A_i Y$, $B_j\mapsto Y^{-1} B_j Z$ and
  $C_k\mapsto Z^{-1} C_k X^{-1}$ for $X,Y,Z\in GL(3)$.  Next, we observe that the antisymmetric
  part of $\lambda$ is invariant under $M\otimes M \otimes M$ with $M\in SL(3)$, the special linear
  group. Hence, restricting to matrices of the form $M=R\oplus\ketbra{2}$, with $R\in SL(2)$, which
  in addition leave $\ket{2,2,2}$ invariant, we also have $(M\otimes M\otimes M)\lambda = \lambda$.
  Thus taking this physical symmetry plus the $Y, Z$ gauge transformation together and restricting for the moment to the
  $2{\times}2{\times}2$ tensor $\widetilde{B}=(B_0, B_1)$, we see that we can apply any operator
  $K_1\otimes K_2 \otimes K_3$ with $K_i\in GL(2)$ to $\widetilde{B}$ without changing
  \eqref{eq:trLambda} if we transform $(A_i)_i$ and $(C_k)_k$ accordingly.
  However $GL(2)^3$ orbits of $2{\times}2{\times}2$-tensors are known explicitly
  \cite{dur2000three}, and we can use this freedom in order to reduce $B_0$ and $B_1$ to seven
  different normal forms, for which we have to obtain a contradiction. In addition to the null tensor and the product
  state, these seven classes encompass the bipartite entanglement between only two parties, the
  W state and the GHZ state. We will now go through all the cases.
  \begin{description}
  \item[null tensor]
       In this case, both $B_0$ and $B_1$ are equal to the zero matrix, which leads for example  to  $\tr(A_i B_1 C_k)=0$, which clearly contradicts \eqref{eq:trLambda}.
  \item[product state]
    In this case, $\tilde B$ can be chosen as $\ket{0}\ket{0}\ket{0}$, which implies $B_0 = \ketbra{0}{0}$ and $B_1$ equal to the zero matrix. Hence, $\tr(A_iB_1C_k) =0$ for all $i,k$ leads to the same contradiction as for the null tensor.
  \item[bipartite entanglement (3 cases)] Depending on the two tensor factors that share the maximally entangled state, $\widetilde{B}$ can be chosen as $\ket{000}+\ket{011}$, $\ket{000}+\ket{101}$ or $\ket{000}+\ket{110}$.
      In the first case $B_0=\Id$ and $B_1=0$, which brings us back to the previous situation. In the remaining two cases $B_0=\ketbra{0}{0}$ and $B_1=\ketbra{0}{1}$ or  $B_1=\ketbra{1}{0}$, respectively.
    \item[GHZ state] In this case, $\widetilde{B}=\ket{000}+\ket{111}$ leading to $B_0=\dyad{0}{0}$ and $B_1 = \dyad{1}{1}$.
  \item[W state] Finally, in this case $\widetilde{B}$ can be chosen as $\ket{000}+\ket{101}+\ket{110}$, giving $B_0 = \dyad{0}{0}$ and $B_1=\ketbra{0}{1}+\ketbra{0}{1}$.
  \end{description}
  In all the cases which we have not immediately discarded, we see that $B_0$ can be chosen as $\dyad{0}{0}$ while
  $B_1$ can either be $\dyad{1}{1}$, $\dyad{1}{0}$, $\dyad{0}{1}$ or $\dyad{0}{1}+\dyad{1}{0}$. We now want to show that neither of
  these cases are possible. We start by decomposing the matrices $A_i$ and $C_k$ as
  \[
    A_i = \dyad{a_i}{0} + \dyad{\tilde a_i}{1}, \quad C_k = \dyad{0}{c_k} + \dyad{1}{\tilde c_k}
  \]
  for vectors $\ket{a_i}, \ket{\tilde a_i}, \ket{c_k}, \ket{\tilde c_k} \in \C^2$.
  Since we have reduced the problem to the case $B_0=\dyad{0}{0}$, we have that
  \[ \tr(A_iB_0C_k) = \braket{c_k}{a_i} = \epsilon_{i,0,k}.\]
  In particular, we have that $\braket{c_1}{a_2} = 1$, $\braket{c_2}{a_1} = -1$, implying that none of these vectors can be the zero vector. Together with $\braket{c_2}{a_2} = 0$ this means that $\linspan\{\ket{a_1}, \ket{a_2}\} = \C^2$, and thus necessarily $\ket{c_0}$ has to be 0, since the trace condition forces it to be orthogonal to both $a_1$ and $a_2$.
  Similarly, we have that $\linspan\{\ket{c_1}\ket{c_2}\} = \C^2$ and that
  $\ket{a_0} = 0$.

  Let us denote the matrix entries of $B_2$ as $b_{i,j}=\tr(B_2 \dyad{j}{i})$ for $i,j=0,1$, and let
  us consider the vectors
  \[
    \ket{a_i'} = b_{0,1}\ket{a_i} + b_{1,1}\ket{\tilde a_i}, \quad \ket{c_k'} =
    \overline{b_{1,0}}\ket{c_k} + \overline{b_{1,1}} \ket{\tilde c_k}.
  \]
  Then it holds that
  \[
    \braket{c_k'}{a_i'} = b_{1,1} \tr(A_iB_2C_k) + (b_{1,0}b_{0,1} - b_{0,0}b_{1,1}) \braket{c_k}{a_i} =
    b_{1,1}(\epsilon_{i,2,k} + \delta_{2,i,k}) - \det(B_2) \epsilon_{i,0,k}.
  \]
  In particular $\braket{c_k'}{a_i'}=0$ for $(i,k) = \{ (0,0), (0,2), (2,0)\}$. Therefore, they
  define an orthogonal representation of $\grph[K]^0_{2,2}$: by \Cref{lemma:graph-repr-4}, either
  $\braket{c_2'}{a_2'}=0$, or either $\ket{a_0'}$ or $\ket{c_0'}$ is zero. We can exclude the latter
  case, since this would imply that either $A_0$ or $C_0$ is zero, which we already know leads to a
  contradiction. Therefore $\braket{c_2'}{a_2'} = b_{1,1} = 0$.
  In the same way, defining
  \[
    \ket{a_i''} = b_{0,1}\ket{a_i} + b_{0,0}\ket{\tilde a_i}, \quad \ket{c_k''} =
    \overline{b_{1,0}}\ket{c_k} + \overline{b_{0,0}} \ket{\tilde c_k},
  \]
  it holds that
  \[
    \braket{c_k''}{a_i''} =  b_{0,0}(\epsilon_{i,2,k} + \delta_{2,i,k}) - \det(B_2) \epsilon_{i,0,k},
  \]
  so we can conclude that also $b_{0,0}=0$.

  We will now consider the four possibilities we have for $B_1$, driving each one of them to a
  contradiction, and therefore showing that no MPS representation of $\lambda$ with bond dimension 2
  is possible.

  \begin{enumerate}
    \item $B_1 = \dyad{1}{0}$: We get a contradiction since $\tr(A_2 B_1 C_0)$ should be $-1$, but
      $B_1 C_0=0$.
     \item $B_1 = \dyad{0}{1}$: We get a contradiction since $\tr(A_0 B_1 C_2)$ should be $1$, but
       $A_0 B_1 =0$.
     \item $B_1=\dyad{0}{1}+\dyad{1}{0}$:
       In this case,
       $\tr(A_iB_1C_k) = \braket{\tilde c_k}{a_i} + \braket{c_k}{\tilde a_i} = \epsilon_{i,1,k}$,
       and in particular $\tr(A_1B_1C_0) = \braket{\tilde c_0}{a_1}$ since $\ket{c_0}=0$.
       From this equation it follows that
       \[
         \tr(A_1B_2C_0) = b_{0,1}\braket{\tilde c_0}{a_1} + b_{1,1}\braket{c_0}{\tilde a_1}
         = b_{0,1} \tr(A_1B_1C_0)  = 0 \neq 1, \]
       so we obtain a contradiction.
     \item $B_1=\ketbra{1}{1}$: We see that $\tr(A_iB_1C_k) = \braket{\tilde c_k}{\tilde a_i} =
      \epsilon_{i,1,k}$, so reasoning in the same way as before we see that $\ket{\tilde a_1} =
      \ket{\tilde c_1} = 0$ and that $\ket{\tilde a_0}$, $\ket{\tilde a_2}$, $\ket{\tilde c_0}$ and $
      \ket{\tilde c_2}$ are non-zero, therefore reducing to the case where
  \begin{align*}
    A_0 &= \dyad{\tilde a_0}{1}, & C_0 &= \dyad{1}{\tilde c_0}, \\
    A_1 &= \dyad{a_1}{0}, & C_1 &=  \dyad{0}{c_1}, \\
    A_2 &= \dyad{a_2}{0} + \dyad{\tilde a_2}{1}, & C_2 &=  \dyad{0}{c_2} + \dyad{1}{\tilde c_2}.
  \end{align*}
  Considering
  \[
    \tr(A_1B_2C_0) = \braket{\tilde c_0}{a_1} b_{0,1} = 1,\quad
    \tr(A_0B_2C_1) = \braket{c_1}{\tilde a_0} b_{1,0} = -1,
  \]
  we obtain that $b_{0,1}, b_{1,0}$ and $\braket{\tilde c_0}{a_1}$, $\braket{c_1}{\tilde a_0}$ are non-zero. On the other hand since $b_{1,1}= 0$ we have that
  \[
    0 = \tr(A_2B_2C_0) = \braket{\tilde c_0}{a_2}b_{0,1},\quad
    0 = \tr(A_0B_2C_2) = - \braket{c_2}{\tilde a_0}b_{1,0},\quad
  \]
   and since $b_{0,1}\neq 0$ and $b_{1,0}\neq 0$ we see that necessarily $\braket{\tilde c_0}{a_2} =
   \braket{c_2}{\tilde a_0}$ = 0.
   Therefore $\ket{a_2}$ is proportional to $\ket{\tilde a_0}$ and similarly $\ket{c_2}$ is
   proportional to $\ket{\tilde c_0}$, and so it follows that
   \begin{align*}
     \frac{\braket{\tilde c_2}{a_2}}{\braket{c_2}{\tilde a_2}} &=
      \frac{\braket{\tilde c_2}{a_2}}{\braket{c_2}{\tilde a_2}} \cdot
     \frac{\braket{c_1}{\tilde a_0}}{\braket{c_1}{\tilde a_0}} \cdot
     \frac{\braket{\tilde c_0}{a_1}}{\braket{\tilde c_0}{a_1}} \\ &=
     \frac{\braket{\tilde c_2}{\tilde a_0}}{\braket{c_2}{a_1}} \cdot
     \frac{\braket{c_1}{a_2}}{\braket{c_1}{\tilde a_0}} \cdot
     \frac{\braket{\tilde c_0}{a_1}}{\braket{\tilde c_0}{\tilde a_2}} = \frac{1}{-1}\cdot
                                                                    \frac{1}{\braket{c_1}{\tilde
                                                                    a_0}} \cdot \frac{\braket{\tilde
                                                                    c_0}{a_1}}{-1} = - \frac{b_{1,0}}{b_{0,1}}.
   \end{align*}
   This leads to a contradiction since
   \[
     \tr(A_2B_2C_2) =b_{0,1}\braket{\tilde c_2}{a_2} + b_{1,0}\braket{c_2}{\tilde a_2} = 0 \neq 1.
   \]
\end{enumerate}
  \end{proof}

\section{Conclusions}
\label{sec:conclusions}

We have shown that analyzing the geometry of entangled states and transformations between entanglement structures provides a  framework for the construction of more efficient tensor network representations. Starting from local improvements on the level of {plaquette} states, we obtain optimized tensor network representations on the entire lattice. We provide two methods to construct such local improvements: restrictions and degenerations.

Using geometrical tools, our main result allows us to lift the local approximate conversion originating from degenerations on the level of \hbox{plaquette} states to an exact representation of the tensor network state on the entire lattice, given as a superposition of tensor network states with smaller bond dimension, the number of which scales linearly in the system size. In addition, our general result gives a prescription of how to leverage this bond dimension reduction in order to reduce the computational cost of computing expectation values. More precisely, we describe a parallel contraction algorithm to compute physical expectation values $\expval{O}{T}$ of the original state as $\sum_i^{2 e F} \gamma_i \expval{O}{V_i}$, where each $V_i$ is given as PEPS of lower bond dimension than $T$.

As an example of application of these techniques, we studied explicitly the RVB state on the kagome lattice. We present two improvement on the representation from \cite{Schuch2012}. The first is obtained by considering bonds of different dimensions, allowing us to arrive at the optimal representation, where two out of three bonds on each triangle can be reduced to bond dimension 2 instead of 3. This leads to saving in the cost of computing contractions (which for the sake of completeness we detailed in \Cref{sec:comp-complexity}). The second improved representation is obtained by considering the more general case of degenerations from the plaquette state \lambdapicture[0.8]: we can then find a border bond dimension $2$ representation of the RVB state, which again is optimal in terms of this effective bond dimension.

More generally, given an entanglement structure $\Phi$ built from locally distributed multi-partite entangled states, our result allows to characterize the variational class given by the set of states obtained by applying local maps $\{A_i(\epsilon)\}_{i=0}^L$ which are
polynomial of degree $e$ in $\epsilon$, and then taking the limit $\epsilon$ to zero. Each state
obtained in this fashion is specified by a polynomial number of parameters. Our main theorem then shows that this gives us access to states which arise as a superposition of a linear number of states represented by $\Phi$, going beyond the states representable by a single tensor network state of this bond dimension.  Nevertheless, their expectation values can be efficiently computed by interpolation.

An interesting question that we leave open for future research is how to optimize efficiently within the set of tensor network states that arise as degenerations, e. g. in order to minimize the energy of a local Hamiltonian. Since degenerations are still given by a local tensor albeit depending polynomially on a free parameter, we expect that a local optimization step along the lines of the known tensor network techniques will be possible. However, one has to carefully take into account the additional constraint for obtaining an honest degeneration and we leave the details of such an optimization procedure for future work.

\clearpage
\section*{Acknowledgments}
We thank Ignacio Cirac, Christian Krumnow and David P\'erez-Garc\'ia for helpful
discussions. M.\,C.\ acknowledges the hospitality of the Center for Theoretical
Physics at MIT, where part of this work was carried out.
\paragraph{Funding information}
We acknowledge financial support from the European Research Council (ERC Grant Agreement
no.~337603 and ERC Grant agreement no.~818761), the Danish Council for Independent Research (Sapere Aude) and
VILLUM FONDEN via the QMATH Centre of Excellence (Grant no.~10059).
A.\,L.\ acknowledges support from the Walter Burke Institute for
Theoretical Physics in the form of the Sherman Fairchild Fellowship as well as support from the
Institute for Quantum Information and Matter (IQIM), an NSF Physics Frontiers Center (NFS Grant
PHY-1733907).
A.\,H.\,W.\ thanks the VILLUM FONDEN for its support with a Villum Young Investigator Grant (Grant No.
25452) and the Humboldt Foundation for its support with a Feodor Lynen
Fellowship.
P.\,V.\ acknowledges support by the National Research, Development and Innovation Fund of Hungary within the Quantum Technology National Excellence Program (Project Nr.~2017-1.2.1-NKP-2017-00001) and via the research grants K124152, KH129601.

\begin{appendix}
\section{Computational complexity of tensor-contractions}\label{sec:comp-complexity}

In this appendix, we derive estimates on the computational cost of exactly and approximately contracting PEPS
networks for the two-dimensional square lattice and the kagome lattice. We will subsequently discuss a specialized contraction strategy for the RVB state from the literature.

\subsection{Exact contraction of the RVB state on the kagome lattice}
 Before discussing contractions of tensor networks with varying bond dimension in the subsequent subsection, let us briefly comment on the computational complexity of exactly contracting the RVB state on the kagome lattice in regards to bond and border bond PEPS representations.
 One strategy employed to contract a PEPS on a lattice is to treat one boundary of the two PEPS layers as an MPS and to view the contraction of the remaining rows of PEPS tensors as the application of matrix product operators to this boundary MPS (see Figure~\ref{fig:contrReg2DLat} and \ref{fig:contrKagOv}). For a PEPS with bond dimension $D$ on the kagome lattice the computational cost of the contraction of a single local tensor into the boundary MPS is given by $\mathcal{O}(\chi^3 D^4) + \mathcal{O}(\chi^2 D^6 d)$ \cite{JordanThesis}, with $d$ the physical dimension and $\chi$ denoting the bond dimension of the boundary MPS. However, since we do not reduce the bond dimension of the boundary MPS $\chi$, we can omit the final SVD, and the relevant scaling is simply $\mathcal{O}(\chi^2 D^6 d)$.

Contracting one full row of local PEPS tensors into the boundary MPS increases its bond dimension by a factor of $D^2$ due to the double layer structure of the network, i.e. $\chi_{i+1} = D^2 \chi_i$. In the case of an exact contraction, this bond dimension is not compressed after each step, and hence $\chi$ grows exponentially with $D^2$. Accordingly, if we consider the computational cost of computing an expectation value of a PEPS with bond dimension $D$ on a $2(L+1)\times 2(L+1)$ lattice and taking into account that we can use a boundary MPS from both sides of the lattice, the overall computational cost is given by
\begin{align}
  2(L+1) \sum_{l=1}^{L}  \mathcal{O}(\chi_l^2 D^6 d) &= \sum_{l=1}^{L}  \mathcal{O}(\chi_0^2 (D^{2l})^2 D^6 d)  = 2(L+1)\mathcal{O}(\frac{D^{4L}-1}{D^4-1} D^{10} d)\;,
\end{align}
%\begin{align}
%  2(L+1) \sum_{l=1}^{L} \mathcal{O}(\chi_l^3 D^4) + \mathcal{O}(\chi_l^2 D^6 d) &= \sum_{l=1}^{L} \mathcal{O}(\chi_0^3 (D^{2l})^3 D^4) + \mathcal{O}(\chi_0^2 (D^{2l})^2 D^6 d)  \\&= 2(L+1)\left(\mathcal{O}(D^{10} \frac{D^{6L}-1}{D^6-1})+ \mathcal{O}(D^{10} \frac{D^{4L}-1}{D^4-1} d) \right)\;,
%\end{align}
with $\chi_0$ the bond dimension of the first boundary MPS. In case of a border PEPS representation with border bond dimension $D$ and error degree $eL$, i.e. linear in the system size as happens for the case of local degenerations of the plaquette tensor, Theorem~\ref{thm:main} implies that we have to contract $2eL+1$ PEPS in order to compute an exact expectation value, leading to a scaling of
\begin{align}
2(L+1)(2eL+1)\mathcal{O}(\frac{D^{4L}-1}{D^4-1}D^{10}  d)\;.
\end{align}
%\begin{align}
%2(L+1)(2eL+1)\left(\mathcal{O}(D^{10} \frac{D^{6L}-1}{D^6-1})+ \mathcal{O}(D^{10} \frac{D^{4L}-1}{D^4-1} d) \right)\;.
%\end{align}
Accordingly, considering the bond dimension $3$ PEPS representation of the RVB state in comparison to the border bond dimension $2$ representation, we obtain a leading scaling of $\mathcal{O}(L\,3^{4L})$ versus a scaling of $\mathcal{O}(L^2 2^{4L})$, which gives an exponential improvement.
\begin{figure}[t]
\centering
%     \tikzexternalenable
 %\tikzsetnextfilename{figure02}
  %\includegraphics[width=0.8\textwidth]{XContrPepsI.tex}
  \includegraphics[width=0.8\textwidth]{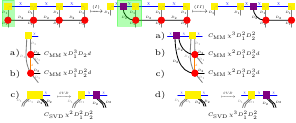}
  %   \tikzexternaldisable
  \caption{Approximate contraction of a PEPS network on the two-dimensional square lattice with the boundary-MPS method.
  The first row shows the initial step at the boundary (I) and the bulk-step (II), which is repeated until the right boundary of the network is reached. For simplicity, only a single layer of the two-layer PEPS-network is shown here, but each red circle in the upper row represents the two local PEPS tensors that have to be contracted along the invisible physical dimension.  In the second row the detailed contractions of both PEPS-layers that are carried out in each step are depicted with their corresponding computational cost.  Lines that terminate in a tensor at a given sub-step ( a)-d)) in a tensor represent the contractions carried out at this point, whereas lines not connected to a tensor at that level correspond to free indices.  In total, the scaling is given by  $ (C_{\text{MM}} + C_{\text{SVD}})\,\chi^3 D_1^2 D_2^2 + 2C_{\text{MM}}\,\chi^2  D_1^3 D_2^3 d$.}\label{fig:contrReg2DLat}
\end{figure}

\subsection{PEPS with varying bond dimension}
We now turn to the contraction of PEPS networks on the kagome and square lattice. In contrast to the results commonly stated in the literature, we will explicitly deal with the case of non-equal bond dimensions with respect to different virtual degrees of freedom and in the case of the kagome lattice also take into account different distributions of the legs in the two layers of the network.
In all cases, we consider a boundary-MPS approach, where the PEPS tensors at the boundary of the network are considered as an MPS of fixed bond dimension $\chi$ to which the internal PEPS tensor regarded as MPOs are applied subsequently. All bounds are based on the estimates $C_{MM} D_1 D_2 D_3$ for the computation of the product of two rectangular matrices of dimensions $D_1{\times} D_2$ and $D_2{\times} D_3$ and $C_{\text{SVD}}\, \chi D_1 D_2 $ for the truncated singular value decomposition (SVD) of a $D_1{\times}D_2$ matrix to its largest $\chi$ singular values \cite{halko2011finding} with $C_{\text{MM}}$ and $C_{\text{SVD}}$ constants.

\noindent\textbf{Two-dimensional square lattice}
Starting from one boundary of the lattice, the next row of the double of the contraction are depicted in \Cref{fig:contrReg2DLat}.  Starting from the left-boundary, the first MPO-tensor (red circle) of the next row is contracted into the boundary-MPS and its bond dimension subsequently reduced to $\chi$ via an SVD (step (I)). The cost of each step in this contraction is indicated in  the second row of \Cref{fig:contrReg2DLat}. In each of the steps a), b) and c), the contractions performed in that step are indicated by lines that terminate in a tensor at that level, all other lines count as free indices. In step I.a) for example the only contraction performed is with respect to the gray line connecting the yellow square and the red circle, whereas the remaining lines (two gray, one black, one orange) are free indices. Hence, this contraction can be seen as a multiplication between a $\chi D_1{\times}D_1$ matrix (yellow square) and a $D_1{\times}D_1 D_2 d$ matrix (red circle) leading to an overall cost of $C_{\text{MM}}\cdot\chi D_1^3 D_2 d$. The two red circles correspond to the two layers of the PEPS network. Hence, the overall cost for contracting the MPO into the boundary-MPS at the boundary is given by
\[
C_{\text{SVD}}\,\chi^2 D_1^2 D_2^2 + C_{\text{MM}}\,\chi D_1^3(D_2^2+D_2)d.
\]
In step (II), the sub-steps b) to d) are basically the same one as the steps a) to c) in step (I),
however we first have to take care of the violet tensor resulting from the SVD performed in sub-step
I.c. The overall computational cost is then given by
\[ (C_{\text{MM}} + C_{\text{SVD}})\,\chi^3 D_1^2 D_2^2 + 2C_{\text{MM}}\,\chi^2
  D_1^3 D_2^3 d\;.
  \]
Because this cost upper bounds the contraction cost at the boundary, cost of contracting each MPO tensor into the boundary-MPS tensor can be upper bounded by
\begin{equation}\label{eq:cost-2d-lat}
 (C_{\text{MM}} + C_{\text{SVD}})\,\chi^3 D_1^2 D_2^2 + 2C_{\text{MM}}\,\chi^2
  D_1^3 D_2^3 d\;.
\end{equation}
which agrees with the estimate for uniform bond
dimension $\mathcal{O}(\chi^2 D^6 d) + \mathcal{O}(\chi^3 D^4)$ found in the literature
\cite{verstraete2004renormalization,verstraete2008matrix,lubasch2014algorithms}.

\begin{figure}[t]
\centering
%\tikzexternalenable
 %\tikzsetnextfilename{figure03}
  %\includegraphics[width=0.8\textwidth]{XContrKagI.tex}
  \includegraphics[width=0.8\textwidth]{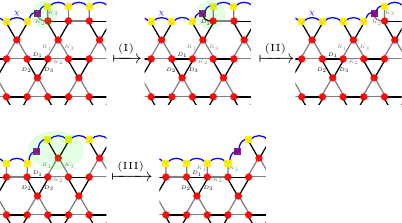}
  %   \tikzexternaldisable
  \caption{Approximate contraction of a PEPS network on the kagome lattice with the boundary-MPS method. Depending on the position of the local tensor to be contracted into the boundary-MPS in the kagome lattice, three different contractions have to be performed. We allow for different bond dimensions for up- ($K_i$) and downwards ($D_i$) pointing triangles.}\label{fig:contrKagOv}
 \end{figure}

\noindent\textbf{Kagome lattice} The situation for the kagome lattice is very similar when compared to the
square two-dimensional lattice except more care has to be taken about how to associate the
local tensors to the boundary-MPS tensors. The procedure we adopt here is depicted in
\Cref{fig:contrKagOv}. In order to make the procedure more transparent, we first split the boundary
vertices at the tip of each triangle into two lattice sites, before we start the contraction
procedure. Fixing the three bond dimensions in each triangle for the full lattice, we can nevertheless distinguish their distribution for upwards ($K_1$, $K_2$, $K_3$) and downwards ($D_1$, $D_2$, $D_3$) pointing triangles.
In comparison to the square two-dimensional lattice, we have to distinguish three different contraction steps, depending on whether we are contracting a tensor on the top right (I), the top left (II) or in the middle (III) of a hexagon.
These three steps are then repeated until the right boundary of the kagome lattice is reached.

Figure~\ref{fig:contrKagStepsI} (I) to (III) depicts the details of these three steps, breaking down every step into the explicit tensor-contractions performed and how expensive
they are in terms of the dimension of the indices of the involved local tensors. In order to
realize improved savings, we allow different distributions of the three bond dimensions in the two triangles for the upper and lower PEPS-layer, indicated by $D_i^\uparrow/D_i^\downarrow$ or $K_i^\uparrow/K_i^\downarrow$, respectively. Taking the maximum over the different computational costs in the three different contractions steps for $\chi^2$ and $\chi^3$ separately, we can upper bound the computational cost of each of all local contractions by
\begin{align*}
  C_{\text{SVD}}\,\chi^3\max(&D_1^\uparrow D_1^\downarrow D_2^\uparrow D_2^\downarrow,\,
   D_3^\uparrow D_3^\downarrow K_2^\uparrow K_2^\downarrow,\,
  K_1^\uparrow K_1^\downarrow K_3^\uparrow K_3^\downarrow)\\
  + C_{\text{MM}}\,\chi^3\max(&K_2^\uparrow K_2^\downarrow  K_3^\uparrow K_3^\downarrow,\,
  D_1^\uparrow D_1^\downarrow K_1^\uparrow K_1^\downarrow,\,
  D_2^\uparrow D_2^\downarrow  D_3^\uparrow D_3^\downarrow)\\
  + C_{\text{MM}}\,\chi^2 d \max(&K_2^\uparrow K_2^\downarrow K_3^\uparrow K_3^\downarrow D_1^\uparrow
    D_2^\uparrow  +  K_2^\downarrow K_3^\downarrow D_1^\uparrow D_1^\downarrow D_2^\uparrow D_2^\downarrow, \\
   &D_1^\uparrow D_1^\downarrow K_1^\uparrow K_1^\downarrow D_3^\uparrow K_2^\uparrow  + D_1^\downarrow K_1^\downarrow D_3^\uparrow D_3^\downarrow K_2^\uparrow  K_2^\downarrow,
  D_2^\uparrow D_2^\downarrow D_3^\uparrow D_3^\downarrow K_1^\uparrow K_3^\uparrow  +
   D_2^\downarrow D_3^\downarrow K_1^\uparrow K_1^\downarrow K_3^\uparrow  K_3^\downarrow).
\end{align*}
In the case, where $D_1=D_3\leq D_2$, choosing the same distribution of the bond dimensions in  both
layers and both types of triangles, i.e. $D_i^\uparrow=D_i^\downarrow$ and $K_i=D_i$ we obtain an
upper bound of
\begin{equation}\label{eq:cost-kagome-lat}
 (C_{\text{SVD}} + C_{\text{MM}})\, \chi^3 D_1^2 D_2^2 +  2 C_{\text{MM}}\, \chi^2 D_1^3 D_2^3 d,
\end{equation}
which has a similar scaling as the square two-dimensional lattice. In the case where all bond dimensions are equal, we arrive at a scaling $\mathcal{O}(\chi^3 D^4) + \mathcal{O}(\chi^2 D^6 d)$ in correspondence with previous results in the literature \cite{JordanThesis}.

\begin{figure}[t]
\centering
%\tikzexternalenable
% \tikzsetnextfilename{figure04}
%  \includegraphics[width=0.3\textwidth]{XContrKagII.tex}
% \tikzsetnextfilename{figure05}
%  \includegraphics[width=0.3\textwidth]{XContrKagIII.tex}
%  \tikzsetnextfilename{figure06}
%   \includegraphics[width=0.3\textwidth]{XContrKagIV.tex}
%     \tikzexternaldisable
       \includegraphics[width=0.3\textwidth]{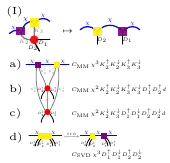}
       \includegraphics[width=0.3\textwidth]{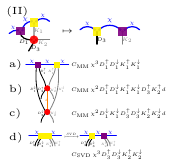}
       \includegraphics[width=0.3\textwidth]{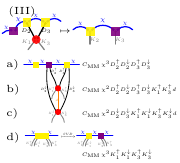}

  \caption{Illustration of the three contraction steps involved in the approximate contraction of a PEPS on the kagome lattice. (I) Details on step (I): The first row depicts the overall contraction step. (a)-(d) show the contractions performed in each sub-step and their
    corresponding computational cost. The superscript $\updownarrow$ indicates, whether a given
    index corresponds to the upper or lower level of the PEPS network. As in the case of the square two-dimensional lattice, lines terminating in a tensor for a given step are contracted, whereas non-terminated lines correspond to free indices of the tensors. (II) Details on step (II): The first row depicts the
     overall contraction step. (a)-(d) show the contractions performed in each sub-step and their
     corresponding computational cost. The superscript $\updownarrow$ indicates, whether a given
     index corresponds to the upper or lower level of the PEPS network. As in the case of the square
     two-dimensional lattice, lines terminating in a tensor for a given step are contracted, whereas non-terminated lines correspond to free indices of the tensors. (III) Details on step (III): The first row depicts the
     overall contraction step. (a)-(d) show the contractions performed in each sub-step and their
     corresponding computational cost. The superscript $\updownarrow$ indicates, whether a given
     index corresponds to the upper or lower level of the PEPS network. As in the case of the square
     two-dimensional lattice, lines terminating in a tensor for a given step are contracted, whereas non-terminated lines correspond to free indices of the tensors.}\label{fig:contrKagStepsI}
 \end{figure}

\subsection{Approximate contraction of the RVB state}

In this section, we apply the estimates on the computational costs of approximately contracting a PEPS on the kagome lattice derived so far in the context of the restrictions and degenerations for the RVB state. As discussed in Section~\ref{sec:LambdNoBd2} considering general restrictions and allowing for unequal bond dimension, we can obtain a PEPS representation of the RVB state, where
two out of three bonds on each triangle of the kagome lattice are reduced to bond dimension 2 instead of 3 (see
\eqref{eq:RVBopRestr}).

According to \eqref{eq:cost-kagome-lat}, the computational cost of contracting a PEPS with bond dimensions satisfy $D_1 = D_3\leq D_2$ around a
triangular plaquette scales as $C_1 \chi^3 D_1^2 D_2^2 + C_2 \chi^2 D_1^3 D_2^3 d$,  where
$C_i$ are constants and $\chi$ denotes the bond dimension of the boundary-MPS. Hence, the optimized tensor network representation of the entanglement structure generated by \lambdapicture[0.8] underlying the RVB state reduces
the prefactor of $\chi^3$ from  $81C_1$ to $36C_1$ and for $\chi^2 d$ from $729C_2$ to $216C_2$ for restrictions.  Note that this improvement applies to the contraction of all tensor networks based on the \lambdapicture[0.8] entanglement structure on the kagome lattice.
The same entanglement structure representing the RVB state is used in \cite{Schuch2012} to construct
a family of quantum states which interpolates between the RVB state and a dimer state, which are
believed to lie in different quantum phases. Since we have improved the PEPS representation of the
entanglement structure behind all these states, the saving we have obtained for the RVB state
applies to all of them.
Note further that, obviously, there are ways to optimize the contraction cost for specific tensor
networks. In \cite{Schuch2012}, for instance, the kagome lattice is first transformed to a square
lattice for which an RVB-specific improved double layer bond dimension is derived. For this particular example, our contraction scheme does not provide an advantage, however we would like to stress that our approach will allow a representation of border bond dimension $2$ for any state obtained as a restriction from the $\Lambda$ entanglement structure.

If we consider the more general case of approximating the plaquette state \lambdapicture[0.8] in terms of degenerations, we can use the border bond dimension $2$ representation of the RVB state (\eqref{eq:RVBdegen}). Employing the parallel contraction algorithm presented in Section~\ref{sec:exStatRep}, allows us to take advantage of this reduction also in the case of the RVB state or any other tensor network state based on the \lambdapicture[0.8] entanglement structure on the kagome lattice. The reduction to border bond dimension $2$ reduces the prefactors for the computational effort for the contraction of each of the $4F$ expectation values $\expval{O}{V_i}$  to $16 C_1$ for $\chi^3$ and $64 C_2$ for $\chi^2d$ as compared to $36C_1$ and $216C_2$ for the unbalanced optimal restriction with bond dimension $(2,2,3)$.

\end{appendix}
%\bibliography{Bibliography}

\end{document}